\documentclass{LMCS}

\def\dOi{9(3:8)2013}
\lmcsheading%
{\dOi}
{1--26}
{}
{}
{Jan.~14, 2013}
{Aug.~30, 2013}
{}

\usepackage{graphicx}
\usepackage[]{amssymb}
\usepackage{mathrsfs}
\usepackage[T1]{fontenc}
\usepackage{tikz}
\usetikzlibrary{decorations,arrows,shapes,automata,backgrounds,fit,calc,petri,patterns}
\usepackage{amsmath}
\usepackage{extarrows}
\usepackage{enumerate,hyperref}

\newcommand{\Sim}{\mathit{sim}}
\newcommand{\Ser}{\mathit{ser}}
\newcommand{\Dep}{\mathit{dep}}
\newcommand{\Ind}{\mathit{ind}}
\newcommand{\Ssm}{\mathit{ssm}}
\newcommand{\Wdp}{\mathit{wdp}}
\newcommand{\Api}{\mathit{sin}}
\newcommand{\Sin}{\mathit{sin}}

\newcommand{\suff}[1]{\mathit{suff_#1}}
\newcommand{\pref}[1]{\mathit{pref_#1}}

\newcommand{\pipe}{\:|\:}

\newcommand{\tracelike}{radical }

\renewcommand{\emptyset}{\varnothing}
\newcommand{\stepalph}{\mathbb{S}}

\newcommand{\StandardNet}[1]    {
     \begin{tikzpicture}[node distance=1.3cm,>=latex',line  width=0.3mm,scale=#1,auto,bend angle=45]
     \tikzstyle{place}=[draw,circle,thick,minimum size=6mm]
     \tikzstyle{transition}=[draw,regular polygon,thick,regular polygon sides=4,minimum size=8mm, inner sep = -2pt]
                              }

\newcommand{\StandardTS}     {
     \begin{tikzpicture}[node distance=0.5cm,>=stealth',bend angle=45,scale=0.5]
                              }

\renewcommand{\put}[3]{
	\node at (#1,#2) {#3};
}

\begin{document}

\title[Algebraic Structure of Combined Traces]{Algebraic Structure of Combined Traces\rsuper*}

\author[{\L}. Mikulski]{{\L}ukasz Mikulski}
\address{
Faculty of Mathematics and Computer Science
\\
Nicolaus Copernicus University
\\ 
Toru{\'n}, Chopina 12/18, Poland, and\vspace{-8 pt}
}
\address{
School of Computing Science
\\
Newcastle University
\\ 
Newcastle upon Tyne, NE1 7RU, U.K.
}
\email{lukasz.mikulski@mat.umk.pl}

\ACMCCS{[{\bf Theory of computation}]:  Models of computation---Concurrency}
\titlecomment{{\lsuper*} A short variant of this paper, without proofs, appeared in the CONCUR 2012 conference proceedings.}

\begin{abstract}
Traces -- and their extension called combined traces (comtraces)
-- are two formal models used in the analysis and verification of concurrent systems.
Both models are based on concepts
originating in the theory of formal languages, and they are able to capture 
the notions of causality and simultaneity of
atomic actions which take place during the process of a system's operation.
The aim of this paper is a transfer to the domain of comtraces 
and developing of some fundamental
notions, which proved to be successful in the theory of traces.
In particular,
we introduce and then apply the notion of indivisible steps,
the lexicographical canonical form of comtraces,
as well as the representation of a comtrace utilising 
its linear projections to binary action subalphabets.
We also provide two algorithms related to the new notions.
Using them, one can solve, in an efficient way, the problem of step sequence
equivalence in the context of comtraces.
One may view our results as a first step towards the development of infinite 
combined traces, as well as recognisable languages of combined traces.
\end{abstract}

\maketitle

\section{Introduction}
The dynamic behaviours of concurrent systems are usually described as
sequences of atomic actions of such systems, which leads to its formal language
semantics. Using this simple approach we cannot express some phenomena, e.g,
concurrency and causality, that are crucial in the process of understanding and
analysing concurrent behaviours of a system. In the case of a particular
operational model, one can consider extending the sequential description
by adding some information about the relevant properties of behaviours.
One can do it by considering sequences of steps of actions and by adding
some causal dependencies between actions.
A well known approach that helps to capture concurrency and causality of a system
are traces~\cite{CarFoa69,Maz77}. 

Consider, for example, the elementary net system with inhibitor 
arcs in Example~\ref{e:ENwithInh}($a$).
We have four actions, $a$, $b$, $c$ and $d$, which may be executed in the
initial marking, and two actions, $e$ and $f$, which need a previous history
of computation to be enabled. Let us focus on action $e$. To enable this action
we need to execute actions $a$ and $c$. We can execute them together or in any 
order. To capture the concurrent behaviour of this computation
we need to identify two sequences of executions -- $ace$ and $cae$. Using
step semantics, which is not necessary in this case, we add also step
sequence $(ac)(e)$ as another possible execution. Traces are sufficient
to deal with such behaviours.

The situation is more complex in the case of action $f$. Now we need three tokens
in the pre-set of the considered action, hence actions $b$, $c$ and $d$ should be
executed before the action $f$. Because of the presence of inhibitors, there is only
one way to execute them sequentially, they should be executed in the order $\mathit{dcbf}$.
Note that $\mathit{bdcf}$ or $\mathit{bcdf}$ are not correct sequences of execution.
There are, however, other possibilities to execute the four actions in the step
semantics. For instance all three actions may be executed simultaneously as
a step containing $b$, $c$ and $d$. This gives $(bcd)(f)$ as our allowed sequence of steps.
Other step sequences are $(d)(bc)(f)$ and $(cd)(b)(f)$. It is important that action $d$ has 
to be executed not later than action $c$, and action $c$ has to be executed not 
later than action $b$. In this case traces are still applicable, but they lose 
some important behavioural information.

Another case is depicted in Example~\ref{e:ENwithInh}($b$). The upper part of
the net is identical to the first case. Here, however, there is a single action
$g$ that waits for tokens in all four middle places. In other words, whole
tuple $(a,b,c,d)$ has to be executed before action $g$. It is easy to see
that because of inhibitors there is no valid sequential execution of the four actions.
After executing one of these actions, one of the remaining becomes disallowed. The only
possible execution is the step sequence $(abcd)(f)$. 
Those two situations cannot be precisely described by traces, we need a more complex
notion that capture ``not later than'' relationship between actions.
To address this issue one can use a natural generalisation of traces called 
combined traces (see \cite{JanKou95}).

\begin{exa}\label{e:ENwithInh}
Two elementary net systems with inhibitor arcs.\\
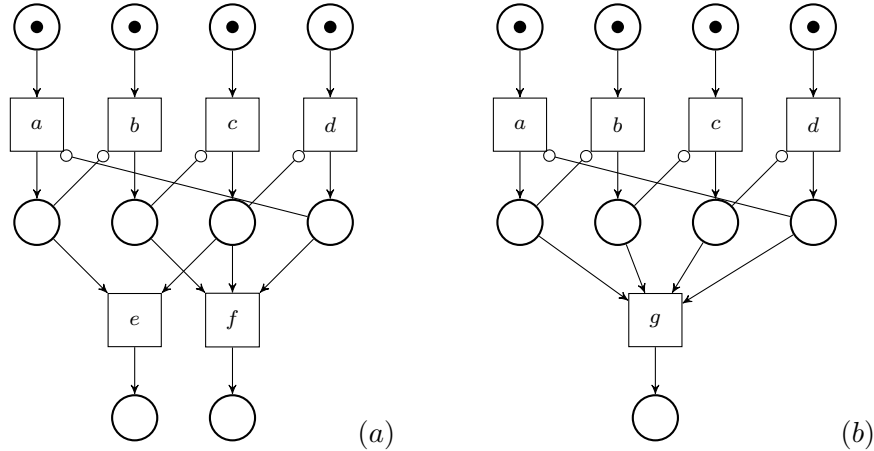
\begin{center}
\begin{tikzpicture}[node distance=1.3cm,>=stealth',bend angle=45,auto]
\tikzstyle{every node}=[font=\scriptsize]
\tikzstyle{place}=[draw,circle,thick,minimum size=6mm]
\tikzstyle{transition}=[draw,regular polygon, regular polygon sides=4,minimum size=10mm]
\node (p1) [place,tokens=1]{};
\node (p2) [place,tokens=1, right of=p1] {};
\node (p3) [place,tokens=1, right of=p2] {};
\node (p4) [place,tokens=1, right of=p3] {};
\node (a7) [transition, below of=p1, yshift=-0.30cm, xshift=-0.05cm, draw=none] {};
\node (a1) [transition, below of=p1] {$a$};
\node (a2) [transition, below of=p2] {$b$};
\node (a3) [transition, below of=p3] {$c$};
\node (a4) [transition, below of=p4] {$d$};
\node (p5) [place, below of=a1] {};
\node (p6) [place, below of=a2] {};
\node (p7) [place, below of=a3] {};
\node (p8) [place, below of=a4] {};
\node (a5) [transition, below of=p6] {$e$};
\node (a6) [transition, below of=p7] {$f$};
\node (p9) [place, below of=a5] {};
\node (p10) [place, below of=a6] {};
\path (p1) edge [->] (a1);
\path (p2) edge [->] (a2);
\path (p3) edge [->] (a3);
\path (p4) edge [->] (a4);
\path (a1) edge [->] (p5);
\path (a2) edge [->] (p6);
\path (a3) edge [->] (p7);
\path (a4) edge [->] (p8);
\path (p5) edge [-o] (a2);
\path (p6) edge [-o] (a3);
\path (p7) edge [-o] (a4);
\path (p8) edge [-o] (a7);
\path (p5) edge [->] (a5);
\path (p7) edge [->] (a5);
\path (p6) edge [->] (a6);
\path (p7) edge [->] (a6);
\path (p8) edge [->] (a6);
\path (a5) edge [->] (p9);
\path (a6) edge [->] (p10);
\end{tikzpicture}($a$)
\hspace{10mm}
\begin{tikzpicture}[node distance=1.3cm,>=stealth',bend angle=45,auto]
\tikzstyle{every node}=[font=\scriptsize]
\tikzstyle{place}=[draw,circle,thick,minimum size=6mm]
\tikzstyle{transition}=[draw,regular polygon, regular polygon sides=4,minimum size=10mm]
\node (p1) [place,tokens=1]{};
\node (p2) [place,tokens=1, right of=p1] {};
\node (p3) [place,tokens=1, right of=p2] {};
\node (p4) [place,tokens=1, right of=p3] {};
\node (a7) [transition, below of=p1, yshift=-0.30cm, xshift=-0.05cm, draw=none] {};
\node (a1) [transition, below of=p1] {$a$};
\node (a2) [transition, below of=p2] {$b$};
\node (a3) [transition, below of=p3] {$c$};
\node (a4) [transition, below of=p4] {$d$};
\node (p5) [place, below of=a1] {};
\node (p6) [place, below of=a2] {};
\node (p7) [place, below of=a3] {};
\node (p8) [place, below of=a4] {};
\node (a5) [transition, below of=p6, xshift=0.5cm] {$g$};
\node (p9) [place, below of=a5] {};
\path (p1) edge [->] (a1);
\path (p2) edge [->] (a2);
\path (p3) edge [->] (a3);
\path (p4) edge [->] (a4);
\path (a1) edge [->] (p5);
\path (a2) edge [->] (p6);
\path (a3) edge [->] (p7);
\path (a4) edge [->] (p8);
\path (p5) edge [-o] (a2);
\path (p6) edge [-o] (a3);
\path (p7) edge [-o] (a4);
\path (p8) edge [-o] (a7);
\path (p5) edge [->] (a5);
\path (p6) edge [->] (a5);
\path (p7) edge [->] (a5);
\path (p8) edge [->] (a5);
\path (a5) edge [->] (p9);
\end{tikzpicture}($b$)
\end{center}
\end{exa}
\vspace{2mm}
\noindent In this paper, we are concerned with the understanding of the algebraic inner 
structure of the combined traces (comtraces in short). 
We start by recalling some standard notions about
formal languages, traces and comtraces. 
In particular, we give the definition of a lexicographical order on step
sequences. We then recall the Foata canonical form of
a comtrace that turns out to be maximal with respect to their order, and
propose another canonical representative - the lexicographical canonical form.
Then, we discuss the phenomenon of indivisibility in the case of comtraces 
and its connections with lexicographical canonical form.
In the following sections, we propose an algebraic representation of a comtrace
based on projections onto sequential subalphabets, and give a nondeterministic
procedure that allows to reconstruct step sequences of the original comtrace.
We also give two strategies of determining such reconstruction, each 
leading to a proper canonical form of a comtrace. 
In the final section, we describe some natural
applications of the algebraic properties developed in this paper, 
and sketch the directions for further research.

The preliminary version of this paper was presented on the CONCUR 2012 conference
(Newcastle, UK) and published in local proceedings. 
The present paper is significantly extended and improved version. 

\section{Preliminaries}
\label{sect-1}

Throughout the paper we use the standard notions of the formal language theory.
In particular, by an \textit{alphabet} we mean a nonempty finite set $\Sigma$, the elements of which
are called \textit{(atomic) actions}.
Finite sequences over  $\Sigma$ are called \textit{words}.
The set of all finite words, including the empty word $\epsilon$, is denoted by $\Sigma^*$.

Let $w=a_1\ldots a_n$ and $v=b_1\ldots b_m$ be two words.
Then 
$$
    w\circ v=wv=a_1\ldots a_nb_1\ldots b_m
$$
is the concatenation of $w$ and $v$.
The alphabet $\mathit{alph}(w)$ of $w$  is the set of all
the actions occurring within $w$, and
$\#_a(w)$ is the number of occurrences of an action $a$ within $w$. 
By $|w|$ we denote the length of word $w$. More generally, for an object $X$, 
whenever the notion of size is clear from the contexts, we denote its size by $|X|$.

Let $w=a_1\ldots a_n$ be a word. We use the notions of prefix and 
suffix of the word $w$. For any $k\leq n$, the
\emph{k-suffix} of $w$, denoted by $\suff{k}(w)$, is a word $a_k\ldots a_n$.
Similarly, the \emph{k-prefix} of $w$, denoted by $\pref{k}(w)$, is the word $a_1\ldots a_k$.

We assume that the alphabet $\Sigma$ is given together with a total 
order $\leq$, called lexicographical order and extend it
to the level of words. Such an order is inherited from the first
actions on which two words being compared differ. 
In the case that one word is a prefix
of another - the former is the smaller one.

The projection onto a binary subalphabet $\{a,b\}$ is the function
$\Pi_{a,b}:\Sigma^*\rightarrow\Sigma^*$ defined as follows:

\[\Pi_{a,b}(cw)=\left\{
\begin{array}{lcl}
c\Pi_{a,b}(w) & \text{ for } & c\in\{a,b\}\\
\Pi_{a,b}(w) & \text{ for } & c\notin\{a,b\}\\
\end{array}
\right.\]
and $\Pi_{a,b}(\epsilon)=\epsilon$. In the same way we define a projection onto a
unary subalphabet $\{a\}$, denoted by $\Pi_{a,a}:\Sigma^*\rightarrow\Sigma^*$.  

The algebra of binary relations over set $X$ (i.e., subsets of $X\times X$)
is equipped with a concatenation operation $\circ$,
where $R_1\circ R_2=\{(x,y)\pipe\exists_{z\in X}\;xR_1y\wedge yR_2z\}$.
The neutral element for $\circ$ is the identity relation $I_X=\{(x,x)\pipe x\in X\}$,
the index $X$ is omitted if it is clear from context.
The $n$-th power of a relation $R$ is defined as 
$R^n=R^{n-1}\circ R$ for all $n\geq 1$, where $R^0=I$.
The transitive closure of $R$ is $R^+=R^1\cup R^2\cup\ldots$, while its
reflexive transitive closure is $R^*=R^0 \cup R^+$. 
Moreover, for a relation $R\subset X\times X$ we define the reverse of $R$ by
$R^{-1}=\{(x,y)\pipe (y,x)\in R\}$, and its symmetric closure by
$R^{sym} = R \cup R^{-1}$. 
We also define the largest equivalence relation contained in the reflexive and
transitive closure of relation $R$ as 
$$R^\circledast=\{(x,y)\pipe xR^*y \wedge yR^*x\}.$$ 
The relation $R\subseteq X\times X$ is called symmetric if $R=R^{-1}$,
reflexive if $I\subseteq R$, irreflexive if $I\cap R=\emptyset$, 
transitive if $R^2\subseteq R$, and acyclic if $R^+$ is irreflexive.
Moreover, for every $Y\subseteq X$ we define the restriction of the relation
$R\subseteq X\times X$ to the set $Y$ by 
\[R|_Y=\{(x,y)\in R\pipe x,y\in Y\}.\]

A \emph{directed acyclic graph} is a pair $\mathit{dag}=(X,R)$,
where $X$ is a finite set and $R$
is an acyclic irreflexive binary relation on $X$.
In a diagrammatical representation, $X$ is the set of vertices while
$R$ the set of arcs.
A directed acyclic graph $\mathit{po}=(X,\prec)$ is a \emph{poset} if the relation
$\prec$ is transitive.
An \emph{upper set} is a nonempty subset $U$ of poset $\mathit{po}=(X,\prec)$ such 
that for every $x\in U$ if $x\prec y$ then $y\in U$.

\subsection{Elementary Net Systems with Inhibitor Arcs}

In this paper we introduce some algebraic properties of combined traces which
are the abstract model that describes causal relationships between executed 
actions of a concurrent system. The underlying structure, which was a motivation
to define combined traces, are elementary net systems with inhibitor arcs.

Formally, the \emph{elementary net system with inhibitor arcs} (or $ENI-system$) is 
a tuple $N=(P,T,F,I,M_0)$, where $P$ and $T$ are two disjoint and finite sets 
of \emph{places} and \emph{transitions} (or \emph{actions}) respectively. Two other 
components, $F\subseteq(P\times T)\cup(T\times P)$ and $I\subseteq P\times T$ are 
relations, called \emph{flow relation} and \emph{inhibition relation}.
These relations describe possible dynamic behaviours of a net, which are
manifested by executing sets of enabled transitions called \emph{steps}. Such an 
execution leads from one set of places (called \emph{marking}) to 
another. The initial marking $M_0\subseteq P$, 
from which the action of a system begins, is the last element of the 
tuple $N$.

Given an ENI-system $N=(P,T,F,I,M_0)$ and $x\in P\cup T$, the \emph{pre-set} 
(set of inputs) of $x$, denoted by $^\bullet x$, 
is defined as $^\bullet x=\{y|(y,x)\in F\}$,
while the \emph{post-set} (set of outputs) of $x$, 
denoted by $x^\bullet$, is defined as $x^\bullet=\{y|(x,y)\in F\}$.
We also use the notion $^\bullet x^\bullet$ for
the union of the post-set and pre-set of $x$, calling it the 
\emph{set of neighbouring places/transitions} (or simply the
\emph{neighbourhood}). 
Moreover, if $x\in T$, the \emph{inh-set} (set of inhibitors) of $x$,
denoted by $^\circ x$, is defined by $^\circ x=\{y|(y,x)\in I\}$.
The set of neighbouring places together with the inh-set forms an
\emph{extended neighbourhood} of an action.

The dot notations are lifted in the usual way to sets of elements.
Hence, by $^\bullet X$ we denote the set $\{y\pipe (y,x)\in F\wedge x\in X\}$,
$X^\bullet=\{y\pipe (x,y)\in F\wedge x\in X\}$, 
and $^\circ X=\{y\pipe (y,x)\in I\wedge x\in X\}$.
Graphically, the places are drawn as circles, transitions as rectangles, elements 
of flow relation as arcs, and elements of inhibition relation as arcs with small
circles as arrowheads. Marked places are depicted by drawing small dot 
called \emph{token} inside.

We say that a step $S=\{t_1,t_2,\ldots,t_n\}$ is \emph{enabled} in marking $M$
if $^\bullet S\subseteq M$, 
$S^\bullet\cap M=\emptyset$, 
$^\circ S\cap M=\emptyset$
and $^\bullet t_i\cap\, ^\bullet t_j=\emptyset$ for any $i\neq j$. 
The \emph{execution} of such a step $S$ leads
from the marking $M$ to the new marking 
$M'=(M\setminus\; ^\bullet S) \;\cup\;S^\bullet$.

An ENI-system with empty inhibition relation, often considered under
the sequential rather than step semantics,
is called an \emph{elementary net system} (or $EN-system$).
\vspace{2.5cm}
\begin{exa}\label{e:ENI}
Consider a system $N=(P,T,F,I,M_0)$ depicted below.\\
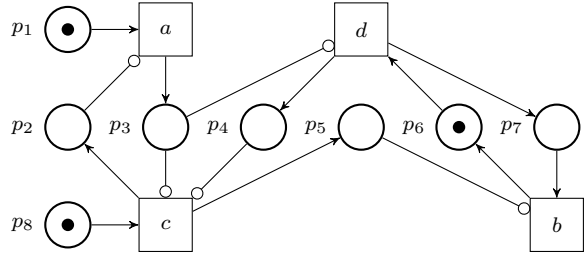
\begin{center}
\begin{tikzpicture}[node distance=1.3cm,>=stealth',bend angle=45,auto]
\tikzstyle{every node}=[font=\scriptsize]
\tikzstyle{place}=[draw,circle,thick,minimum size=6mm]
\tikzstyle{transition}=[draw,regular polygon, regular polygon sides=4,minimum size=10mm]
\node (p1) [place,tokens=1, label=left:$p_1$]{};
\node (p2) [place, below of=p1, label=left:$p_2$] {};
\node (p8) [place, below of=p2, tokens=1, label=left:$p_8$]{};
\node (a1) [transition, right of=p1] {$a$};
\node (p3) [place, below of=a1, label=left:$p_3$] {};
\node (a3) [transition, below of=p3] {$c$};
\node (p4) [place, right of=p3, label=left:$p_4$] {};
\node (p5) [place, right of=p4, label=left:$p_5$] {};
\node (p6) [place,tokens=1, right of=p5, label=left:$p_6$] {};
\node (p7) [place, right of=p6, label=left:$p_7$] {};
\node (a4) [transition, above of=p5] {$d$};
\node (a2) [transition, below of=p7] {$b$};
\path (p1) edge [->] (a1);
\path (p8) edge [->] (a3);
\path (p2) edge [-o] (a1);
\path (a1) edge [->] (p3);
\path (a3) edge [->] (p2);
\path (p3) edge [-o] (a3);
\path (p3) edge [-o] (a4);
\path (a4) edge [->] (p4);
\path (p4) edge [-o] (a3);
\path (a3) edge [->] (p5);
\path (p5) edge [-o] (a2);
\path (a2) edge [->] (p6);
\path (p6) edge [->] (a4);
\path (a4) edge [->] (p7);
\path (p7) edge [->] (a2);
\end{tikzpicture}
\end{center}
\vspace{2mm}

\noindent The set of places has eight elements (from $p_1$ to $p_8$), the set 
of transitions has four elements ($T=\{a,b,c,d\}$). 
In the initial marking, three places are marked -- $(p_1,p_6,p_8)$.
Therefore, seven steps -- including $(a)$, $(d)$ and $(ad)$ --  are enabled. Note that 
after executing transition $d$, transition $a$ remains enabled, however, this does
not hold in the opposite direction, i.e. after executing transition $a$ there is
a token in place $p_3$ and transition $d$ is no more enabled.
\qed
\end{exa}

\subsection{Traces}
\label{sect-2}

In this section we recall well-known notion of traces 
(see \cite{DieRoz95,Maz77,Mik08}). Traces are an abstract model describing 
causal relationships between executed actions in, for example EN-systems. 
They capture
independence, hence the possibility to be executed in any order (and also together)
for some actions. Structurally, pairs of actions with disjoint sets of 
neighbouring places are in the independence relation.

A \textit{concurrent alphabet} is a pair $\Psi=(\Sigma,\Ind)$, where $\Sigma$
is an alphabet and $\Ind\subseteq \Sigma\times\Sigma$ is an
irreflexive and symmetric \textit{independence} relation.
The corresponding \textit{dependence} relation
is given by $\Dep=(\Sigma\times\Sigma)\setminus \Ind$.

A concurrent alphabet
$\Psi$ defines an equivalence relation $\equiv^{\Sigma}_\Psi$
identifying words which differ only by the ordering of independent actions.
Two words, $w,v\in\Sigma^*$, satisfy $w\equiv^{\Sigma}_\Psi v$
if there exists a finite sequence of commutations of
adjacent independent actions transforming $w$ into $v$.
More precisely, $\equiv^{\Sigma}_\Psi$ is a binary relation over $\Sigma^*$
which is the reflexive and transitive closure of the
relation $\sim^{\Sigma}_\Psi$ such that
$w\sim^{\Sigma}_\Psi v$ if there are $u,z\in\Sigma^*$ and $(a,b)\in \Ind$
satisfying $w=uabz$ and $v=ubaz$.

Equivalence classes of $\equiv^{\Sigma}_\Psi$ are
called \textit{(Mazurkiewicz) traces}
and the trace containing a given word $w$ is denoted by $[w]$.
The set of all traces over $\Psi$ is denoted by $\Sigma^*/_{\equiv_\Psi^\Sigma}$,
and the pair $(\Sigma^*/_{\equiv_\Psi^\Sigma},\circ)$ is a (trace) monoid,
where $\tau\circ\tau'=[w\circ w']$, for any words $w\in\tau$ and $w'\in\tau'$,
is the concatenation operation for
traces. Note that trace concatenation is well-defined as $[w\circ w']=[v\circ v']$,
for all $w,v\in\tau$ and $w',v'\in\tau'$.
Similarly, for every trace $\tau=[w]$ and every action $a\in \Sigma$, we can define
\[
\begin{array}{lclclclclcl}
    \mathit{alph}(\tau)
    &=&
    \mathit{alph}(w)
    & ~~~~~~
    & \#_a(\tau)
    &=&
    \#_a(w).
\end{array}
\]

Projections onto unary and binary dependent subalphabets 
(i.e. $\{a,b\}\subseteq\Sigma$ such that
$(a,b)\in\Dep$) are invariants for traces (see \cite{Mik08}). 
It is possible to formulate the trace equivalence in terms of projections. Two words
$u,w\in\Sigma^*$ are in relation $\equiv_\Psi^\Sigma$ if and only if
\[\forall_{(a,b)\in\Dep}\;\Pi_{a,b}(u)=\Pi_{a,b}(w).\]
Following \cite{Mik08}, we define the \emph{projection representation} of 
$\tau$ as a function $\Pi_\tau:\Dep\rightarrow\Sigma^*$, where
$\Pi_\tau(a,b) = \Pi_{a,b}(\tau)$.

\begin{exa}
\label{e:traces}
    Consider a concurrent alphabet $\Psi$ with four actions
    $\Sigma = \{a, b, c, d\}$ together with a dependence relation $\Dep$ given by:
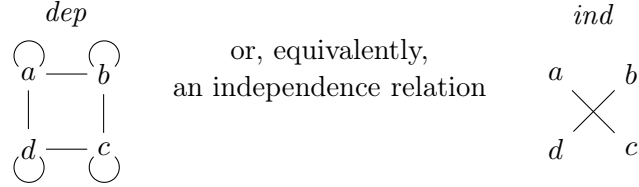
\begin{center}
                \begin{tikzpicture}[baseline=-0.5cm]
                       \node (n1) {$a$};
                       \node [xshift=0.5cm,yshift=0.8cm] {$\Dep$};
                       \node (n2) [right of=n1] {$b$};
                       \node (n3) [below of=n2] {$c$};
                       \node (n4) [left of=n3] {$d$};
            \draw (n1) [xshift=0.15cm,yshift=0.1cm] arc (-45:225:0.2cm);
            \draw (n2) [xshift=0.15cm,yshift=0.1cm] arc (-45:225:0.2cm);
            \draw (n3) [xshift=0.15cm,yshift=-0.1cm] arc (45:-225:0.2cm);
            \draw (n4) [xshift=0.15cm,yshift=-0.1cm] arc (45:-225:0.2cm);
                       \draw (n1) -- (n2) -- (n3) -- (n4) -- (n1);
               \end{tikzpicture}
               \hspace{0.3cm}
               \begin{tikzpicture}
								\node (n0) {};
                        \node (n1) [below of=n0] {or, equivalently,};
                        \node (n2) [below of=n1, yshift=5mm] {an independence relation};
               \end{tikzpicture}
               \hspace{0.3cm}
               \begin{tikzpicture}[baseline=-0.5cm]
                       \node (n1) {$a$};
                       \node [xshift=0.5cm,yshift=0.8cm] {$\Ind$};
                       \node (n2) [right of=n1] {$b$};
                       \node (n3) [below of=n2] {$c$};
                       \node (n4) [left of=n3] {$d$};
                       \draw (n1) -- (n3);
                       \draw (n2) -- (n4);
               \end{tikzpicture}
\end{center}
\vspace{5mm}
\noindent
    Then $w=abbaacd\equiv^{\Sigma}_\Psi abb caad$.

\noindent
    The projection representation of a trace $\tau=[w]$ is
\[
\begin{array}{llllr}
~~~~~~\Pi_{a,a}(\tau)=aaa \; & \;~~ \Pi_{b,b}(\tau)=bb \; & 
\;~~\Pi_{c,c}(\tau)=c \; & \;~~ \Pi_{d,d}(\tau)=d & \\
~~~~~~\Pi_{a,b}(\tau)=abbaa\; & \;~~ \Pi_{a,d}(\tau)=aaad\; &
\;~~\Pi_{b,c}(\tau)=bbc\; & \;~~ \Pi_{c,d}(\tau)=cd & ~~~~~~
\end{array}
\]
\qed\end{exa}

A word $w\in\Sigma^*$ is in \emph{Foata canonical form} (see \cite{DieMet97})
w.r.t.\ the dependence relation  $\Dep$ and 
a lexicographical order $\leq$ on $\Sigma$, if $w=w_1\ldots w_n$ ($n\geq 0$), 
where each
$w_i$ is a nonempty word such that:
\begin{iteMize}{$\bullet$}
\item
   $alph(w_i)$ is pairwise independent and $w_i$
   minimal w.r.t. lexicographical order $\leq$ among $[w_i]$

\item
   for each $i>1$ and action $a$ occurring in $w_i$,
   there exists action $b$ occurring in $w_{i-1}$ such that
   $(a,b)\in \Dep$.
\end{iteMize}

\noindent Another canonical (normal) form of a trace that one may consider is the 
\emph{lexicographical canonical form} (see also \cite{DieMet97}). 
It is based only on the lexicographical order and
is defined as the least representative of a trace with respect
to the lexicographical ordering.
The intuition behind the Foata canonical form is that it groups actions into maximally
concurrent steps, while the lexicographical canonical form is very useful in some
combinatorial approaches (see~\cite{MikPiaSmy11}).
Each trace contains exactly one sequence in the Foata canonical form, and exactly
one sequence in the lexicographical canonical form. 
It may happen that the two versions of canonical form coincide.

\subsection{Step Traces}\label{s:stepTraces}
Let us lift the notion of traces from the sequential semantics discussed above
to the step semantics. 
Instead of identifying sequences of actions over alphabet $\Sigma$, 
we will identify sequences of sets of actions, called steps. 
We demand that a step should consist of mutually independent actions only.

For a given concurrent alphabet $\Psi=(\Sigma,\Ind)$ we define a set $\stepalph_\Psi$
of all nonempty subsets $A\subseteq\Sigma$ such that for all $a,b\in A$ we have
$a\neq b\Rightarrow(a,b)\in\Ind$. 
If the concurrent alphabet $\Psi$ is clear from the context, we would
write $\stepalph$ instead of $\stepalph_\Psi$.
To avoid confusion with the well-established operation of
concatenating sets in formal languages theory, we follow Diekert (\cite{DieMet97}) and denote a step 
containing actions $a$ and $b$ by $(ab)$ rather then $\{a,b\}$, etc.
Finite sequences in $\stepalph^*$, including the empty one $\lambda=(\epsilon)$,
are called \textit{step sequences}.

We now lift a number of notions and notations introduced for words to the
level of step sequences.
In what follows, $\Psi=(\Sigma,\Ind)$ is a \emph{fixed} concurrent 
alphabet.
Let $w=A_1\ldots A_n$ and $v=B_1\ldots B_m$ be two step sequences.
Then $w\circ v=wv=A_1\ldots A_nB_1\ldots B_m$ is the concatenation of $w$ and 
$v$.
The alphabet $\mathit{alph}(w)$ of $w$  comprises all actions occurring within~$w$,
and $\#_a(w)$ is the number of occurrences of an action $a$ within $w$.
Moreover, we define the \emph{step alphabet} $Alph(w)\subseteq\stepalph$ of 
a step sequence $w$ as the set of all steps occurring in $w$.

Both independence and dependence relations may be extended to the case
of steps. 
Two steps $A,B\in\stepalph$ are independent if and only if 
$A\times B\subseteq\Ind$, otherwise they are dependent.
We not only allow to commute, but also to join/split pairs of independent steps.
In fact, the commutation of two independent steps may be composed as two
join/split operations.
More precisely, $\equiv^\stepalph_\Psi$ is a binary relation over $\stepalph^*$
which is the reflexive, symmetric and transitive closure of the relation 
$\sim^\stepalph_\Psi$ such that $w\sim^\stepalph_\Psi v$ if there are
$u,z\in\stepalph^*$ and $A,B\in\stepalph$ satisfying $w=uABz$,
$v=u(A\cup B)z$, and $A\times B\subseteq\Ind$. 
Note that $A\cap B=\emptyset$, since $\Ind$ is irreflexive. 
Equivalence classes of $\equiv^\stepalph_\Psi$ are called \emph{step traces} 
(see \cite{Vog91}). 
The trace containing a step sequence $w$ is denoted by $[w]$, while set of all 
step traces -- by $\stepalph^*/_{\equiv^\stepalph_\Psi}$. 
Step traces are a conservative extension of sequential traces.
To justify this statement we prove

\begin{prop}\label{p:stepTraceSwaping}
Let $\Psi=(\Sigma,\Ind)$ be a concurrent alphabet, and $w,v$ two step sequences
over $\Psi$. If $w=uABz$ and $v=uBAz$, where $u,z\in\stepalph^*$ and 
$A\times B\subseteq\Ind$ then $w\equiv^\stepalph_\Psi v$.
\end{prop} 
\begin{proof}
Directly from the definition, both $w$ and $v$ are in the
relation~$\sim^\stepalph_\Psi$ with $y=u(A\cup B)z$.
Since $\equiv^\stepalph_\Psi$ is the reflexive, symmetric and transitive closure of
$\sim^\stepalph_\Psi$, we have $w\equiv^\stepalph_\Psi y$ and 
$v\equiv^\stepalph_\Psi y$, so also $w\equiv^\stepalph_\Psi v$.
\end{proof}

We define two operations which help
to move from step semantics into sequential semantics and vice versa.
Let $A\in\stepalph$ be a step and $\leq$ be a total order on $\Sigma$.
Using the relation $\leq$ we define $min(A)$, the
minimal representative of a step $\emptyset\neq A\in\stepalph$ 
as the minimal action in $A$ with respect to $\leq$.
Note that $min(\emptyset)$ is not defined.
We define the \emph{lexicographical linearization} of step $A$ as 

\[
lex(A)=\left\{
\begin{array}{lcl}
\epsilon & \text{ for } & A=\emptyset\\
min(A)lex(A\setminus min(A)) & \text{ for } & A\neq\emptyset.\\
\end{array}
\right.
\]
We extend the operation $lex$ to step sequences and sets of step sequences
in the usual way:
\begin{align*}
& lex(A_1A_2\ldots A_n)=lex(A_1)lex(A_2)\ldots lex(A_n),\\
& lex(X)=\{lex(w)\pipe w\in X\}.
\end{align*}
As a reverse operation, we define a \emph{singletonization} of an action $a$ 
by $sstep(a)=\{a\}$ and extend it to the case of sequences by 
$sstep(a_1\ldots a_n)=sstep(a_1)\ldots sstep(a_n)$.

Since no two dependent actions may occur in the same step, we can easily
lift the notion of projections onto unary and binary dependent subalphabets 
to the case of step sequences being representatives of step traces:
\[\Pi_{a,b}(Aw)=\left\{
\begin{array}{lcl}
a\Pi_{a,b}(w) & \text{ for } & a\in A\\
b\Pi_{a,b}(w) & \text{ for } & b\in A\\
\Pi_{a,b}(w) & \text{ for } & \{a,b\}\cap A=\emptyset\\
\end{array}
\right.\]
and $\Pi_{a,b}(\lambda)=\epsilon.$ 
Note that both actions $a$ and $b$ can not simultaneously be in $A$,
since they are dependent.

\begin{prop}\label{p:stepTraceLex}
Let $\Psi=(\Sigma,\Ind)$ be a concurrent alphabet, 
and $w,u\in\stepalph^*$ two step sequences over $\Psi$.
Then $w\equiv^\stepalph_{\Psi} u$ if and only if $lex(w)\equiv^\Sigma_{\Psi} lex(u)$.
\end{prop}
\begin{proof}
Since every action in sequential semantics can be treated as a singleton
step in step semantics, the implication 
$lex(w)\equiv^\Sigma_{\Psi} lex(u) \Rightarrow w\equiv^\stepalph_{\Psi} u$ follows
directly from Proposition~\ref{p:stepTraceSwaping}. 
Therefore we need to prove that 
$w\equiv^\stepalph_{\Psi} u\Rightarrow lex(w)\equiv^\Sigma_{\Psi} lex(u)$.

Recalling the definition of step traces, it is sufficient to show that for 
$A,B\in\stepalph$ such that $A\times B\subseteq \Ind$ we have
$lex(A)lex(B)\equiv^\Sigma_{\Psi}lex(A\cup B)$.
We make use of the projection formulation for sequential trace equivalence.
Since $A\times B\subseteq\Ind$ we have $A\cap B=\emptyset$.
Moreover $A\cup B$ is a step, hence
for every dependent pair $(a,b)$ we have $\Pi_{a,b}(A\cup B)$ empty or equal to
a single action (also in the degenerated case $a=b$). 

Therefore for every $(a,b)\in\Dep$ we have either
\[\Pi_{a,b}(A\cup B)=\Pi_{a,b}(A) \text{ when }
\{a,b\}\cap A\neq\emptyset\wedge\{a,b\}\cap B=\emptyset,\] 
or 
\[\Pi_{a,b}(A\cup B)=\Pi_{a,b}(B) \text{ when }
\{a,b\}\cap B\neq\emptyset\wedge\{a,b\}\cap A=\emptyset.\]
Hence $\Pi_{a,b}(A\cup B)=\Pi_{a,b}(A)\Pi_{a,b}(B)$ and 
$lex(A)lex(B)\equiv^\Sigma_{\Psi}lex(A\cup B)$.
\end{proof}

\begin{prop}\label{p:stepTraceFixPoints1}
Let $w$ be a step sequence over a concurrent alphabet $\Psi$. 
Then the sequences of singletons are fixpoints of the function $sstep\circ lex$, i.e.
\[sstep\circ lex(w)=w ~~\Longleftrightarrow~~ \forall_{A\in Alph(w)}\; |A|=1.\]
\end{prop}
\proof\hfill

\noindent$\Longrightarrow:$\\
Let $w\in\stepalph^*$ and $A\in Alph(w)$ be such that $|A|>1$. 
Without loss of generality we may assume that $A$ is the first step in $w$.
Then 
\[lex(w)=lex(Aw')=lex(A)lex(w')=min(A)lex(A\setminus min(A))lex(w'),\]
hence
\[sstep\circ lex(w)=sstep(min(A))sstep\circ lex((A\setminus min(A))w').\]
As a result we get that the first step in $sstep\circ lex(w)$ is a singleton, 
which is in contradiction with the assumption $|A|>1$.
Hence $\forall_{A\in Alph(w)}\;|A|=1$.

~

\noindent$\Longleftarrow:$\\
The second implication is straightforward, since $|A|=1$ implies $A=\{a\}$ and
$a=min(A)$, so $lex(A)=a$ and $sstep\circ lex(A)=sstep(a)=\{a\}=A$.
Let $w=A_1\ldots A_n$. Directly from the definitions
\[sstep\circ lex(w)=sstep\circ lex(A_1)\ldots sstep\circ lex(A_n)=A_1\ldots A_n=w.\eqno{\qEd}\]

\begin{prop}\label{p:stepTraceFixPoints2}
Let $u$ be a sequence over a concurrent alphabet $\Psi$.
Then $lex\circ sstep(u)=u$.
\end{prop}
\begin{proof}
Let $u=a_1\ldots a_n$ and $A_i=\{a_i\}$. 
Analogously to the proof of Proposition \ref{p:stepTraceFixPoints1}, we have
$a_i=lex(A_i)$ and $sstep(a_i)=A_i$, so 
\begin{multline*}
lex\circ sstep(u)=lex\circ sstep(a_1\ldots a_n)
=lex(sstep(a_1)\ldots sstep(a_n))\\
=lex(A_1\ldots A_n)
=lex(A_1)\ldots lex(A_n)=a_1\ldots a_n=u.
\end{multline*}
\end{proof}

\noindent The pair $(\stepalph^*/_{\equiv^\stepalph_\Psi},\circ)$ is a (step trace) monoid,
where $\tau\circ\tau'=[w\circ w']$, for any
step sequences $w\in\tau$ and $w'\in\tau'$.
Step trace concatenation is
well-defined as $[w\circ w']=[v\circ v']$, for all $w,v\in\tau$ and 
$w',v'\in\tau'$.
A step trace $\tau$ is a prefix of a step trace $\tau'$ if
there is a step trace $\tau''$ such that $\tau\circ\tau''=\tau'$.
As in the case of sequential traces, for every step trace $\tau$ and every 
$a\in \Sigma$, we can define
$\mathit{alph}(\tau)=\mathit{alph}(w)$ and
$\#_a(\tau)=\#_a(w)$, where $w$ is any step sequence belonging to $\tau$.
The situation with the step alphabet, as it is not an invariant of a step trace, is a
bit more complex: we define $Alph(\tau)=\bigcup_{w\in\tau}Alph(w)$.

\begin{thm}\label{t:prReprStTr}
Let $\Psi$ be a concurrent alphabet, and $w,u\in \stepalph^*$.
Then $w\equiv^{\stepalph}_\Psi u$ if and only if 
$\forall_{(a,b)\in\Dep}\;\Pi_{a,b}(w)=\Pi_{a,b}(u)$.
\end{thm}
\begin{proof}
The proof follows directly from Proposition~\ref{p:stepTraceLex} 
and the projection based definition of trace equivalence (in sequential semantics).
\end{proof}

Next, we give the canonical (normal) form of a step trace which essentially
captures a greedy, maximally concurrent, execution
of the actions occurring in the step trace
conforming to the independence relations.
A step sequence $w = A_1\ldots A_n\in\stepalph^*$
is in \emph{Foata canonical form} if, for each $i\leq n$,
whenever $Av\equiv^\stepalph_\Psi A_i\ldots A_n$
for some $A\in\stepalph$ and $v\in\stepalph^*$,
then $A\subseteq A_i$.
One can see that all suffixes and all prefixes of a step sequence in 
Foata canonical form are also in Foata canonical form, and that
each step trace comprises a unique step sequence in Foata canonical form.
Note that the following statement holds:
\begin{prop}\label{p:foata4steptraces}
Let $\Psi$ be a concurrent alphabet.
A step sequence $w = A_1\ldots A_n\in\stepalph^*$
is in \emph{Foata canonical form} if and only if 
for every $i<n$, there is no $\emptyset\neq A\subseteq A_{i+1}$
such that $A_i\times A\subseteq \Ind$.
\end{prop}
\begin{proof}~

\noindent$\Longrightarrow:$\\
Let $w = A_1\ldots A_n$ be in Foata canonical form. 
Suppose that there are $i<n$ and $\emptyset\neq A\subseteq A_i$ such that
$A\times A_i\subseteq\Ind$.
Then, since for every $A_k$ we have $A_k\times A_k\subseteq(\Ind\cup I)$,
\[A_iA_{i+1}\ldots A_n\equiv^\stepalph_\Psi 
  A_iA(A_{i+1}\setminus A)A_{i+2}\ldots A_n\equiv^\stepalph_\Psi
  (A_i\cup A)(A_{i+1}\setminus A)A_{i+2}\ldots A_n,\]
and by $A\neq\emptyset$ we obtain that $w$ is not in Foata canonical form.
Hence there are no such $i<n$ and $A\subseteq A_{i+1}$.

~

\noindent$\Longleftarrow:$\\
Let $w = A_1\ldots A_n$ and for every $i<n$, 
there is no $\emptyset\neq A\subseteq A_{i+1}$
such that $A_i\times A\subseteq \Ind$. 
Assume moreover that $Av\equiv^\stepalph_\Psi A_i\ldots A_n$ 
and let $B=A_i\setminus A$.
Suppose that $B$ is not empty and let $j$ be the least index such that 
$B\cap A_j\neq\emptyset$.
By Theorem \ref{t:prReprStTr} not only such $j$ exist, but also we get that 
$B\times A_k\subseteq\Ind$ for every $i<k<j$.
Hence $B\cap A_j$ is a nonempty step contained in $A_j$ and 
$(B\cap A_j)\times A_{j-1}\subseteq\Ind$, 
which gives a contradiction with the assumptions, 
and so $B$ has to be empty, which ends the proof. 
\end{proof}

We can also distinguish one of the representatives of $\tau$, built 
from singletons.
Note that such step sequences may be considered as sequences over $\Sigma$
and compared using lexicographical order $\leq$.
Similarly to the case of sequential traces, we call the least (with respect to 
the order $\leq$) singleton based representative of a step trace $\tau$ its
\emph{lexicographical canonical form}.

Canonical forms of sequential and step traces connect those two worlds.
More precisely, the following hold:

\begin{thm}\label{t:lex4TrAndSTr}
Let $\Psi=(\Sigma,\Ind)$ be a concurrent alphabet.
Then a step sequence $w$ is in lexicographical canonical form
if and only if 
all $A\in Alph(w)$ are singletons 
and
the sequence $lex(w)$ is in lexicographical canonical form.
\end{thm}
\begin{proof}
According to Proposition \ref{p:stepTraceFixPoints1},
it is only an equivalent reformulation of the definition.
\end{proof}

\begin{thm}\label{t:foataSTr2Tr}
Let $\Psi=(\Sigma,\Ind)$ be a concurrent alphabet.
If a step sequence $w=A_1A_2\ldots A_k$ is in Foata canonical form
then 
the sequence $u=lex(w)$ is in Foata canonical form.
\end{thm}
\begin{proof}
We have to prove that sequences $lex(A_1),\ldots,lex(A_k)$ satisfy the conditions
from the definition of Foata canonical form in the case of sequential trace.
The elements of every $A_i$ are pairwise independent since $A_i$ is a step.
Let us suppose that there exist $1<i\leq k$ and $a\in A_i$ such that $a$ is 
independent with every action $b$ from $A_{i-1}$. Let $A=\{a\}$.
Then $A\neq\emptyset$, $A\subseteq A_{i}$, and $A_{i-1}\times A\subseteq\Ind$.
From the definition of Foata canonical form in the case of step traces
we have that $w$ is not in Foata canonical form.
Hence, $u$ is indeed in Foata canonical form.
\end{proof}

\begin{thm}\label{t:foataTr2STr}
Let $\Psi=(\Sigma,\Ind)$ be a concurrent alphabet.
If a sequence $u=a_1a_2\ldots a_n$ is in Foata canonical form
then 
there exists a step sequence $w=A_1A_2\ldots A_k$ in Foata canonical form
such that $u=lex(w)$.
\end{thm}
\begin{proof}
From the definition of Foata canonical form, we know that there exist sequences
$u_1,\dots, u_k$ such that the elements of $alph(u_i)$ are pairwise independent for
every $0<i\leq k$.
Hence, for every $0<i\leq k$ we have that $alph(w_i)$ is a step over $\Psi$.
Let $A_i=alph(w_i)$.
Suppose that $w$ is not in Foata canonical form.
Then, there exist nonempty $A\in\stepalph$ and $0<i<k$
such that $A\subseteq A_{i+1}$ and $A_i\times A\subseteq\Ind$.
Let $a\in A$. 
Since $A_i\times A\subseteq\Ind$ there is no $b\in A$ dependent with $a$. 
Hence $u$ is not in Foata canonical form.
This contradicts the assumptions and proves the theorem.
\end{proof}

We conclude this subsection by formulating and proving a result that establishes
a relationship between two semantics in which we can consider traces:

\begin{thm}\label{t:stepEQsequential}
Let $\sigma\in\Sigma^*/_{\equiv^\Sigma_\Psi}$ be a trace (in sequential semantics).
Then there exists a unique step trace $\tau$ such that $lex(\tau)=\sigma$.
\end{thm} 
\begin{proof}
Theorems \ref{t:foataSTr2Tr} and \ref{t:foataTr2STr} allow us to associate 
sequential trace $\sigma$ with a trace $\tau$ using their Foata canonical forms.

Let $w$ be a lexicographical canonical form of $\tau$.
By Theorems \ref{t:lex4TrAndSTr} and \ref{t:prReprStTr}, and definition of 
sequential traces based on projections, we get that $lex(w)\in \sigma$
and lexicographical canonical forms of $\sigma$ and $\tau$ also overlaps.
By Proposition \ref{p:stepTraceLex} we conclude that $lex(\tau)\subseteq \sigma$.

Let $X=\{sstep(u)\pipe u\in\sigma\}$.
By Propositions \ref{p:stepTraceFixPoints2} and \ref{p:stepTraceLex} 
we get that $X\subseteq\tau$ and $lex(X)=\sigma$.
Hence $\sigma=lex(X)\subseteq lex(\tau)$ which ends the proof. 
\end{proof}

\subsection{Comtraces}
Whereas traces are satisfactory to describe the concurrent behaviour of 
{EN-systems}, they are not sufficient to capture the behaviour of systems with 
inhibitor arcs. 
To deal with such systems, we recall the notion of combined traces 
(see \cite{JanKou95}).

A  \emph{comtrace alphabet} is a
triple $\Theta=(\Sigma,\Sim,\Ser)$, where
$\Sigma$ is an arbitrary alphabet and
$\Ser\subseteq\Sim\subseteq\Sigma\times\Sigma$ are two
relations, respectively called
\emph{serialisability} and \emph{simultaneity}; it is assumed
that  $\Sim$ is irreflexive and symmetric.
Intuitively, if $(a,b)\in\Sim$
then $a$ and $b$ may occur
simultaneously, whereas $(a,b)\in \Ser$
means that in such a case $a$
may also occur before $b$ (with both executions being equivalent).
The set of all (potential) steps over $\Theta$,
or \emph{step alphabet}, is then defined as the set $\stepalph_\Theta$ comprising all
nonempty sets of actions $A\subseteq\Sigma$
such that $(a,b)\in\Sim$, for all distinct $a,b\in A$.
If the comtrace alphabet $\Theta$ is clear from the context, 
we would write $\stepalph$ instead of $\stepalph_\Theta$.

The \emph{comtrace congruence} over
$\Theta$, denoted by $\equiv_\Theta$, is the reflexive, symmetric and transitive 
closure of the relation $\sim_\Theta\subseteq\stepalph^*\times\stepalph^*$ such 
that $w\sim_\Theta v$ if there are $u,z\in\stepalph^*$ and $A,B\in\stepalph$
satisfying $w=uABz$, $v=u(A\cup B)z$ and $A\times B\subseteq \Ser$.
Note that $A\cap B=\emptyset$ as $\Ser$ is irreflexive.

Equivalence classes of the relation $\equiv_\Theta$
are called \textit{comtraces} (see~\cite{JanKleKou11}),
and the comtrace containing a given step sequence $w$ is denoted by $[w]$.
The set of all comtraces is denoted by $\stepalph^*/_{\equiv_\Theta}$,
and the pair $(\stepalph^*/_{\equiv_\Theta},\circ)$ is a (comtrace) monoid,
where $\tau\circ\tau'=[w\circ w']$, for any
step sequences $w\in\tau$ and $w'\in\tau'$.
Comtrace concatenation is
well-defined as $[w\circ w']=[v\circ v']$, for all $w,v\in\tau$ and 
$w',v'\in\tau'$.
A comtrace $\tau$ is a prefix of a comtrace $\tau'$ if
there is a comtrace $\tau''$ such that $\tau\circ\tau''=\tau'$.
As in the case of step traces, for every comtrace $\tau$ and every $a\in \Sigma$, we 
can define
$\mathit{alph}(\tau)=\mathit{alph}(w)$ and
$\#_a(\tau)=\#_a(w)$, where $w$ is any step sequence belonging to $\tau$.
Moreover, $Alph(\tau)=\bigcup_{w\in\tau}Alph(w)$.

Next, we give the canonical form of a comtrace which essentially
captures a greedy, maximally concurrent, execution
of the actions occurring in the comtrace
conforming to the simultaneity and serialisability relations.
A step sequence $w = A_1\ldots A_n\in\stepalph^*$
is in \emph{Foata canonical form} if, for each $i\leq n$,
whenever $Av\equiv_\Theta A_i\ldots A_k$
for some $A\in\stepalph$ and $v\in\stepalph^*$,
then $A\subseteq A_i$.
This canonical form of a comtrace is extensively discussed in \cite{JanLe11}.
One can see that all suffixes and all prefixes of step sequence in Foata canonical 
form are also in Foata canonical form, and that
each comtrace comprises a unique step sequence in Foata canonical form.

Note that an alternative (equivalent) definition of normal
form requires that, for every $i<k$, there is no $\emptyset\neq A\subseteq A_{i+1}$
such that $A_i\times A\subseteq \Ser$ and $A\times (A_{i+1}{\setminus}A)\subseteq \Ser$.
Moreover, in the cases of sequential and step traces we define two canonical forms.
The first is, as in the case of comtraces, Foata canonical form, 
while the latter is called lexicographical. 
Both of those canonical forms prove to be very elegant and useful theoretical tool
(as an example see prove of Theorem \ref{t:stepEQsequential}).
In the next section we define the lexicographical canonical form of a comtrace.
It is one of the main notions introduced and utilised in this paper.
But previously, let us discuss in detail direct relationships between atomic actions.

\subsection{Relations between actions}

In our discussion, we use a number of relations
capturing semantically meaningful relationships between individual actions (see also \cite{MikKou11}):
\begin{iteMize}{$\bullet$}
\item
Dependence $\Dep=(\Sigma\times\Sigma)\setminus\Sim$,
and independence $\Ind=\Ser\cap\Ser^{-1}$.\\
Both relations have their counterparts in trace theory, and
we denote them in the same way.
If two actions are dependent then they never occur in a common step.  
Two actions are
independent if they can be executed 
in any order  as well as simultaneously (as $\Ser\subseteq\Sim$).

\item 
Semi-independence $\Api = \Sim\setminus\Ser$.\\
In contrast to the situation found in traces, dependence and independence do not describe
all possible relationships between individual actions in comtraces.
The remaining ones are called, due to the possibility of occurring together without
being fully independent, semi-independent actions. 
Semi-independent actions may be further divided into symmetric and antisymmetric parts:

\begin{iteMize}{$-$}
\item
Strong simultaneity $\Ssm=\Sim\setminus(\Ser\cup\Ser^{-1})=\Sin\setminus\Ser^{-1}$. \\
If two actions are strongly simultaneous then may occur simultaneously but cannot be serialised at all.
This means that two occurrences of strongly simultaneous actions which appear together
in a step sequence $w$ would appear together in every step sequence belonging to 
the comtrace $[w]$.

\item
Weak dependence $\Wdp=\Ser^{-1}\setminus\Ser=\Sin\setminus\Sin^{-1}$.\\
Two actions are weakly dependent if they can be
serialised only in one way.
This means that for any two actions $(a,b)\in\Wdp$,
if their occurrences appear in the order `$a$ followed by $b$'
then they behave like completely dependent actions,
while appearing in the order `$b$ followed by $a$' allows one to equivalently 
execute (if there are no other obstacles) a step $(ab)$.
\end{iteMize}
\end{iteMize}

\noindent The main motivation to define all those classes was to capture the essence of
the interplay between single atomic elements of concurrent systems modelled using
comtraces. 
As a result we achieve the projection representation defined later.

Similarly to the case of simultaneity and serialisability,
each of proposed relations can be described semantically by specific
relationships between pre-sets, post-sets and inh-sets of pairs of actions.
Note that if the set of neighbouring places of two actions overlaps,
then those places are automatically considered as dependent 
(like in the case of traces and EN-systems).
The main role in the further partition is played by the extended neighbourhoods.
To capture dependence we have to add (to the overlapping of strict neighbourhoods)
the situation when one action has an input place that is simultaneously 
an inhibitor for the other.

The intersections of post-sets and inh-sets of two different actions are
significant if they are not dependent.
Namely, if they are totally disjoint, which means that their extended neighbourhoods
are disjoint, those two actions are independent.
Remaining situations correspond to the cases when two action have disjoint 
neighbourhoods as well as disjoint pre-sets and inh-sets, but still overlapping 
extended neighbourhoods and are captured by the semi-independence relation.

Note that in the favourable circumstances both of them might be executable
(like in the case of independence), but the execution of one of them may disable 
the execution of the other. 
If $(^\circ b \cap\;a^\bullet)\neq\emptyset
\wedge\;(^\circ a \cap\;b^\bullet)=\emptyset$ then after executing $b$ we can 
immediately execute $a$, but not vice versa. 
While if both $^\circ b \cap\;a^\bullet$ and $^\circ a \cap\;b^\bullet$ are nonempty then we cannot split simultaneous execution of $a$ and $b$.
The following table gives a straightforward description of all seven
relations for ENI-systems.

\vspace{5mm}
\begin{center}
\begin{tabular}{|l|c|l|}
\hline
& & \\[-3mm]
simultaneity & $(a,b)\in\Sim$ & 
$^\bullet a^\bullet \cap\;^\bullet b^\bullet=\emptyset \wedge\; 
(^\circ a \cap\;^\bullet b)\cup\;(^\circ b \cap\;^\bullet a)=\emptyset$\\[1mm]
serialisability & $(a,b)\in\Ser$ &
$(a,b)\in\Sim \wedge a^\bullet\cap\;(^\bullet b\cup \,^\circ b)=\emptyset$\\[1mm]
\hline
& & \\[-3mm]
dependence & $(a,b)\in\Dep$ & 
$^\bullet a^\bullet \cap\;^\bullet b^\bullet\neq\emptyset \vee\; 
(^\circ a \cap\;^\bullet b)\cup\;(^\circ b \cap\;^\bullet a)\neq\emptyset$\\[1mm]
independence & $(a,b)\in\Ind$ & 
$(a,b)\notin\Dep\wedge\;
(^\circ a \cap\;b^\bullet)\cup\;(^\circ b \cap\;a^\bullet)=\emptyset$\\[1mm]
semi-independence & $(a,b)\in\Sin$ & 
$(a,b)\notin\Dep\wedge\;
(^\circ b \cap\;a^\bullet)\neq\emptyset$\\[1mm]
strong simultaneity & $(a,b)\in\Ssm$ & 
$(a,b)\notin\Dep\wedge\;
(^\circ b \cap\;a^\bullet)\neq\emptyset
\wedge\;(^\circ a \cap\;b^\bullet)\neq\emptyset$\\[1mm]
weak dependence & $(a,b)\in\Wdp$ & 
$(a,b)\notin\Dep\wedge\;
(^\circ b \cap\;a^\bullet)\neq\emptyset
\wedge\;(^\circ a \cap\;b^\bullet)=\emptyset$\\[1mm]
\hline
\end{tabular}
\end{center}
\vspace{5mm}

\begin{exa}
\label{e:comtrace}
    Consider a comtrace alphabet $\Theta$ for ENI-system $N$ from 
       Example~\ref{e:ENI}.
    The simultaneity and serialisability
    relations are given by:\\

	\begin{center}
		$\Sim = $
		\begin{tikzpicture}[baseline=-0.5cm, scale=2]
			\node (n1) {$a$};
			\node (n2) [right of=n1] {$b$};
			\node (n3) [below of=n2] {$c$};
			\node (n4) [left of=n3] {$d$};
			\draw (n1) -- (n2) -- (n3) -- (n1) -- (n4) -- (n3);
		\end{tikzpicture}
 ~~~~
		\hspace{0.5cm}
		$\Ser = $
		\begin{tikzpicture}[baseline=-0.5cm, scale=2]
			\node (n1) {$a$};
			\node (n2) [right of=n1] {$b$};
			\node (n3) [below of=n2] {$c$};
			\node (n4) [left of=n3] {$d$};
			\path (n4) edge [->] (n1);
			\path (n3) edge [->] (n4);
			\path (n2) edge [->] (n3);
			\path (n1) edge [->,out=30,in=150] (n2);
			\path (n1) edge [<-,out=-30,in=-150] (n2);
		\end{tikzpicture}
 \\
	\end{center}

\noindent
In the net $N$ we have a pair of independent actions $(a,b)$.
Note that their extended neighbourhoods are disjoint.
The only pair of different and dependent actions is $(b,d)$.
The reason for their dependency is the non-emptiness of their neighbourhoods.
All the remaining pairs of different actions are semi-independent.
Only one of them, namely $(a,c)$, is strongly simultaneous.
Note that the post place of one of these actions is an inhibitor place of
another, forming in the net graph a special kind of cycle.
Similar behaviour (post place which is simultaneously inhibitor place),
may be observed in the remaining cases, namely for pairs $(a,d)$, $(d,c)$ and $(c,b)$.
However, we have there an asymmetric situation and those pairs of actions are
weakly dependent.
The five derived relations on actions are as follows:\\
\newsavebox\myindrel
  \sbox{\myindrel}{%
  $~\Ind = $
    \begin{tikzpicture}[remember picture, baseline=-0.5cm, scale=2]
			\node (n1) {$a$};
			\node (n2) [right of=n1] {$b$};
			\node (n3) [below of=n2] {$c$};
			\node (n4) [left of=n3] {$d$};
			\draw (n1) -- (n2);
    \end{tikzpicture}%
  }
  
\newsavebox\mydeprel
  \sbox{\mydeprel}{%
  		$\Dep = $
		\begin{tikzpicture}[baseline=-0.5cm, scale=2]
			\node (n1) {$a$};
			\node (n2) [right of=n1] {$b$};
			\node (n3) [below of=n2] {$c$};
			\node (n4) [left of=n3] {$d$};
            \draw (n1) [xshift=0.075cm,yshift=0.05cm] arc (-45:225:0.1cm);
            \draw (n2) [xshift=0.075cm,yshift=0.05cm] arc (-45:225:0.1cm);
            \draw (n3) [xshift=0.075cm,yshift=-0.05cm] arc (45:-225:0.1cm);
            \draw (n4) [xshift=0.075cm,yshift=-0.05cm] arc (45:-225:0.1cm);
			\draw (n2) -- (n4);
		\end{tikzpicture}%
  }

\newsavebox\mysinrel
  \sbox{\mysinrel}{%
  $\Api = $  
		\begin{tikzpicture}[baseline=-0.5cm, scale=2]
			\node (n1) {$a$};
			\node (n2) [right of=n1] {$b$};
			\node (n3) [below of=n2] {$c$};
			\node (n4) [left of=n3] {$d$};
			\path (n1) edge [->] (n4);
			\path (n3) edge [->] (n2);
			\path (n4) edge [->] (n3);
			\path (n1) edge [->,out=-20,in=110] (n3);
			\path (n1) edge [<-,out=-70,in=160] (n3);
    \end{tikzpicture}%
  }

\newsavebox\myssmrel
  \sbox{\myssmrel}{%
  $\Ssm = $
		\begin{tikzpicture}[baseline=-0.5cm, scale=2]
			\node (n1) {$a$};
			\node (n2) [right of=n1] {$b$};
			\node (n3) [below of=n2] {$c$};
			\node (n4) [left of=n3] {$d$};
			\draw (n1) -- (n3);
    \end{tikzpicture}%
  }
  
\newsavebox\mywdprel
  \sbox{\mywdprel}{%
  $\Wdp = $
		\begin{tikzpicture}[baseline=-0.5cm, scale=2]
			\node (n1) {$a$};
			\node (n2) [right of=n1] {$b$};
			\node (n3) [below of=n2] {$c$};
			\node (n4) [left of=n3] {$d$};
			\path (n1) edge [->] (n4);
			\path (n3) edge [->] (n2);
			\path (n4) edge [->] (n3);
    \end{tikzpicture}%
  }

\begin{center}
\begin{tikzpicture}
\node[inner sep=3pt] (R1) {\usebox{\myindrel}};
\node[inner sep=3pt, right of=R1, xshift=3cm] (R2) {\usebox{\mysinrel}};
\node[inner sep=3pt, right of=R2, xshift=3cm] (R3) {\usebox{\mydeprel}};
\node[inner sep=3pt, below of=R1, xshift=2.3cm,yshift=-2cm] (R4) {\usebox{\myssmrel}};
\node[inner sep=3pt, right of=R4, xshift=2.5cm] (R5) {\usebox{\mywdprel}};
\draw[color=lightgray, -stealth, line width=2mm, postaction={draw, line width=0.2cm, shorten >=1cm, -}] (R2) -- (R4);
\draw[color=lightgray, -stealth, line width=2mm, postaction={draw, line width=0.2cm, shorten >=1cm, -}] (R2) -- (R5);
\end{tikzpicture}
\end{center}

\noindent
The combined trace of one of the possible executions in the net $N$ is
$\tau=\{w,v,u,z\}$, where:
\[
\begin{array}{lcl}
    w
    &=&
    (d)(ab)
    \\[1mm]
    v
    &=&
    (d)(a)(b)
    \\[1mm]
    u
    &=&
    (ad)(b)
    \\[1mm]
    z
    &=&
    (d)(b)(a)\;.
\end{array}
\]

\noindent
Moreover, $u$ is a step sequence in Foata canonical form.
\qed\end{exa}

\section{Lexicographical canonical form}\mbox\\

We extend the order on actions to the case of steps (sets of actions). Let
$A,B\in\stepalph$ be two steps. If the size of $A$ is smaller then the size
of $B$ then $A\widehat{\leq} B$. If the sizes are equal, 
$A\widehat{\leq} B$ if $A=B$ or $A\neq B$ and
$min(A\setminus B)\leq min(B\setminus A)$. 
In this way, $(\stepalph,\widehat{\leq})$ becomes a totally ordered set.

Using the order $\widehat{\leq}$ we can define \emph{lexicographical order} on
step sequences in the usual way. 
The \emph{lexicographical canonical form} of 
a comtrace $\tau$, denoted by $minlex(\tau)$, is the least (with respect to 
the lexicographical order $\widehat{\leq}$) step sequence contained in the comtrace.
Note that, in contrast to the Foata canonical form, the lexicographical
canonical form captures one of the most sequential executions of a comtrace. 
Hence the two canonical forms lie on the opposite sides of 
the concurrent/sequential spectrum of behaviours.
Note that the step sequence $v$ from Example \ref{e:comtrace} is 
in lexicographical canonical form (assuming $a<b<c<d$).

\begin{thm}\label{t:extremalforms}
For a given comtrace $\tau$, its Foata canonical form is 
the $\widehat{\leq}$-greatest, 
and its lexicographical canonical form is 
the $\widehat{\leq}$-least, 
step sequence contained in $\tau$.
\end{thm}
\proof
The lexicographical canonical form is the $\widehat{\leq}$-least step sequence contained in $\tau$
directly from the definition. We need to prove that Foata canonical form is greater
than any other step sequence contained in $\tau$.

Let $u=A_1\ldots A_n, v=B_1\ldots B_m$, $u\neq v$, $u\equiv_\Theta v$, and
$u$ be in Foata canonical form. Moreover, 
let $i=min\{k \pipe k\leq n \;\wedge\; A_k\neq B_k\}$. Note that
such a number $i$ exists, since $u\neq v$ and $u\equiv_\Theta v$ so one
sequence cannot be a prefix of another. 

We have $A_1\ldots A_{k-1}=B_1\ldots B_{k-1}$,
so directly form the definition of Foata canonical form
$B_k\neq A_k\wedge B_k\subseteq A_k$.
Since $B_k\widehat{\leq}A_k$, we have $B_k\ldots B_m\widehat{\leq}A_k\ldots A_n$,
and $v\widehat{\leq}u$. 
\qed

\subsection{Indivisible steps and sequences}

The structure and semantics of relations $\Sim$ and $\Ser$ mean that some
actions have to appear simultaneously in every step sequence contained in 
a comtrace (in other word, they cannot be separated according to the comtrace congruence). 
A very good example of such actions are those in the $\Ssm$ relation.
The strong simultaneity, however, does not exhaust all situations when actions 
are ``glued'' together in a permanent manner. 
Such a behaviour was used in \cite{MikKou11} to form so called folded actions.
It is also worth to observe that the notion of indivisible steps was discussed, 
in the case of step traces with auto-concurrency, in \cite{Vog91}.
In this section, we discuss the phenomenon of the indivisibility 
(in the case of comtraces) in depth.

Let us consider a step $A\in\stepalph$ and a relation 
$\equiv_A\subseteq A\times A$, such that, for all $a,b\in A$, we have 
$a\equiv_A b$ if $(a,b)\in(\Api|_A)^\circledast$.
Intuitively, the relation $\equiv_A$ joins actions that can be executed
simultaneously, but cannot be executed in a sequential way (see Example \ref{e:indiv}).
Note that, for arbitrary step $A$, the relation $\equiv_A$ is an equivalence 
relation. 

We say that a step $A\in\stepalph$ is \emph{indivisible} if  
$\forall_{a,b\in A}\;a\equiv_A b$. The set of all indivisible steps 
is denoted by
$\widehat{\stepalph}$. By $indiv(\tau)$ we denote the set of 
all step sequences contained in a comtrace $\tau$ and built with indivisible 
steps only. 

\begin{exa}\label{e:indiv}
Let us recall the comtrace alphabet from Example \ref{e:comtrace} and 
the relations $\Sim$ and $\Api$, which are crucial in determining indivisible steps.
	\begin{center}
		$\Sim = $
		\begin{tikzpicture}[baseline=-0.5cm, scale=2]
			\node (n1) {$a$};
			\node (n2) [right of=n1] {$b$};
			\node (n3) [below of=n2] {$c$};
			\node (n4) [left of=n3] {$d$};
			\draw (n1) -- (n2) -- (n3) -- (n1) -- (n4) -- (n3);
		\end{tikzpicture}
 ~~~~
		\hspace{1.5cm}
		$\Api = $
		\begin{tikzpicture}[baseline=-0.5cm, scale=2]
			\node (n1) {$a$};
			\node (n2) [right of=n1] {$b$};
			\node (n3) [below of=n2] {$c$};
			\node (n4) [left of=n3] {$d$};
			\path (n1) edge [->] (n4);
			\path (n4) edge [->] (n3);
			\path (n3) edge [->] (n2);
			\path (n1) edge [->,out=-20,in=110] (n3);
			\path (n1) edge [<-,out=-70,in=160] (n3);
		\end{tikzpicture}
\end{center}

\noindent The set of all possible steps is 
$\stepalph=\{(a),(b),(c),(d),(ab),(ac),(ad),(bc),(cd),(abc),(acd)\}$,
while the set of all indivisible steps is 
$\widehat{\stepalph}=\{(a),(b),(c),(d),(ac),(acd)\}$.
Note that step $A=(abc)$ is divided by the relation $\equiv_A$ into two
indivisible steps $B=(b)$ and $C=(ac)$ and step $B$ occurs not later than
step $C$, while step $D=(ab)$ is divided by the relation $\equiv_D$ into two,
completely independent, indivisible steps $(a)$ and $(b)$.

Moreover, there are only two sequences of indivisible steps contained in the comtrace $\tau$ 
which is defined in Example \ref{e:comtrace}. These two sequences are $v=(d)(a)(b)$ and $z=(d)(b)(a)$.
\qed
\end{exa}

Intuitively, we can treat the indivisible step sequences
belonging to $indiv(\tau)$ as classical sequences over the alphabet $\widehat{\stepalph}$.
Hence we define two complementary relations over this alphabet, 
the independence relation $\widehat{\Ind}$ and 
the dependence relation $\widehat{\Dep}$. 
We say that two indivisible steps $A$ and $B$ are \emph{independent} if 
$A\times B\subseteq\Ind=\Ser\cap\Ser^{-1}$; otherwise two indivisible steps are \emph{dependent}.

\begin{prop}
All steps contained in the lexicographical canonical form of a comtrace are 
indivisible $(minlex(\tau)\in indiv(\tau))$.
\end{prop}
\begin{proof}
Suppose, to the contrary, that $minlex(\tau)=uAv$ contains a 
non-indivisible step $A$.
We conclude from Lemma \ref{l:division} that for two disjoint steps $B$ and 
$C$ we have a step sequence $uBCv\in\tau$ which is different from the 
step sequence $minlex(\tau)$.
Since $B\subseteq A$ and $A\neq B$ we have $uBCv\;\widehat{\leq}\;uAv$ so we found
a step sequence contained in $\tau$ that is lexicographically smaller than 
$minlex(\tau)$, which contradicts our assumption. Hence all steps contained
in $minlex(\tau)$ are indivisible.
\end{proof}

Recall the $lex$ operator defined in Section \ref{s:stepTraces}.
It allows us to translate a step sequence to a sequence of actions,
and was very helpful in dealing with step traces.
In the case of comtraces, however, it has rather narrower application. 
Therefore, we define the \emph{split operator} that translates arbitrary step 
sequences to step sequences of indivisible steps as
\mbox{ }$\widehat{ }:\stepalph^*\rightarrow\widehat{\stepalph}^*$ as
\[\widehat{(A_1\ldots A_n)}=\widehat{A_1}\ldots\widehat{A_n}
=minlex(A_1)\ldots minlex(A_n).\]
The following facts justify an observation that 
the split operator does not lead beyond the comtrace, see Proposition \ref{p:leadout}.

\begin{lem}\label{l:division}
Let $A\in\stepalph\setminus\widehat{\stepalph}$ be a step that is not
indivisible. Then there exist two steps, $B$ and
$C$, such that $A\sim_\Theta BC$. Moreover, $A/_{\equiv_A}=B/_{\equiv_B}\cup C/_{\equiv_C}$.
\end{lem}
\proof
Since $A$ is not indivisible, the relation $\equiv_A$ divides $A$ into at 
least two equivalence classes. 
In the following proof we choose an indivisible step, to play a role of $B$.
However, at first we separate a special subset of $A$, denoted by $D$.
One can think about $D$ as a set of elements from $A$, 
which form a minimal layer in the graph of the relation $(\Sin|_A)^*$.

Let $D$ be the set of all actions $b\in A$ such that, 
$\forall_{a\in A}\;(b,a)\in(\Sin|_A)^*\Rightarrow a\in[b]_{\equiv_A}$. 
Suppose that
$D$ is empty. Let us take any $b_1\in A$. Then, by $D=\emptyset$,
there exists $b_2\in A$ such that $b_2\notin[b_1]_{\equiv_A}$ and 
$(b_1,b_2)\in(\Sin|_A)^*$. 
Continuing in this way, we can construct an infinite sequence
of actions $b_i\in A$ such that, for all $i$,
$b_{i+1}\notin[b_1]_{\equiv_A} \wedge (b_i,b_{i+1})\in(\Sin|_A)^*$.

Since $A$ is finite, the elements contained in this sequence have
to repeat. Let $b_n=b_m$ and $n<m$. Since $(\Sin|_A)^*$ is transitive we have 
$(b_{n+1},b_n)\in(\Sin|_A)^*$ and $(b_m,b_{n+1})\in(\Sin|_A)^*$, so $b_{n+1}\in[b_n]_{\equiv_A}$
which contradicts the assumption. Hence $D$ is not empty.

Let $d$ be an arbitrary element from $D$ and $B=[d]_{\equiv_A}$.
$A$ is not indivisible, hence $A\neq B$. 
Moreover, directly from the construction of the set $D$, $B\subseteq D$.
Let $b\in B$ and $a\in A\setminus B$. 
From the definition of $D$ we have that $(b,a)\notin(\Sin|_A)^*$, 
so $(b,a)\notin\Wdp$ and $(b,a)\notin\Ssm$. 
We also have $(a,b)\notin\Dep$ since $a$ and $b$ are both contained in $A$. 
This gives
\[\forall_{a\in A\setminus B} \forall_{b\in B}\;(a,b)\in\Wdp\vee (a,b)\in\Ind.\] 
Hence
\[\forall_{a\in A\setminus B} \forall_{b\in B}\;(b,a)\in\Ser\]
and finally $A\sim_\Theta B(A\setminus B)$.

~

\noindent It remains to be proven that  
$A/_{\equiv_A}=B/_{\equiv_B}\cup C/_{\equiv_C}$.
According to the definition of the comtrace equivalence, $A\sim_\Theta BC$
implies that $B\times C\subseteq \Ser$. It means that for every pair of
actions $b\in B$ and $c\in C$ we have $(b,c)\notin \Api$. 
Hence for every $a,b\in A$ we have 
$a\equiv_A b \wedge a\in B \Rightarrow b\in B$
and 
$a\equiv_A b \wedge a\in C \Rightarrow b\in C$.

It means that the graphs of the relation $\Api$ restricted to steps $B$ and $C$
not only are vertex induced parts of the graph of the relation $\Api$ restricted
to the step $A$, but also are a division of this graph
(i.e., the union of strongly connected components of graphs 
$\Api|_B$ and $\Api|_C$ is equal to the set of strongly 
connected components of the graph $\Api|_A$), which end the proof. 
\qed

\begin{prop}
Let $\tau$ be a comtrace over $\Theta$ and $A\in Alph(\tau)$. Then
$$A/_{\equiv_A}\subseteq Alph(\tau).$$ 
\end{prop}
\begin{proof}
Since $A\in Alph(\tau)$, there exists $w,u\in\stepalph$ such that $wAu\in\tau$.
Applying Lemma \ref{l:division} we can construct the step sequence 
$A_1\ldots A_n$ composed of indivisible steps only and 
equivalent to step sequence consisting of $A$ only.
Moreover, $A/_{\equiv_A}=\bigcup_{i=1\ldots n} A_i$ and $wA_1\ldots A_nu\in\tau$.
As a result we get that $A_i\in Alph(\tau)$, 
hence $A/_{\equiv_A}\subseteq Alph(\tau)$.
\end{proof}

\begin{thm}\label{t:trace4comtrace}
Let $\tau$ be a comtrace. The set $indiv(\tau)$ is a trace 
(with sequential semantic) over the concurrent 
alphabet $(\widehat{\stepalph},\widehat{\Dep})$.
\end{thm}
\begin{proof}
To prove the statement of the theorem it is sufficient to show two facts.
Firstly, we need to prove that relation $\widehat{\Ind}$ is symmetric and 
irreflexive. Secondly, we need to argue that by the repeated transposing
of two subsequent and independent actions (in fact indivisible steps) 
we can reach any of other elements of the set $indiv(\tau)$ and cannot go beyond 
this set.

We start from the first statement. By the definition of $\Ser$ the relation 
$\Ind=\Ser\cap\Ser^{-1}$ is symmetric and irreflexive. Since two 
indivisible steps $A$ and $B$ are in relation $\widehat{\Ind}$ if all pairs
of actions $(a,b)\in A\times B$ are independent, we conclude that the relation
$\widehat{\Ind}$ is also symmetric and irreflexive.

Let $w=uABv$ be a step sequence from $indiv(\tau)$ and $(A,B)\in\widehat{\Ind}$.
By the definition of the $\widehat{\Ind}$ relation we have 
$AB\sim_\Theta C$ and $BA\sim_\Theta C$, where $C=A\cup B$. 
Therefore $uABv\equiv_\Theta uBAv$ and the set 
$indiv(\tau)$ is equal to its own trace closure. The last needed statement
follows from Lemma \ref{l:division} (about indivisibility of indivisible steps). 

Let us suppose that there are two
comtrace equivalent step sequences $u$ and $v$  belonging to $indiv(\tau)$
that are not trace equivalent. Hence they differ in at least one
projection to a binary dependent subalphabet, so
there are two occurrences
of indivisible steps $A$ and $B$ that appear in the two different orders
and are dependent ($(A,B)\in\widehat{\Dep}$).
Let $A$ precede $B$ in the step sequence $u$, 
and $B$ precede $A$ in the step sequence $v$. From the definition of
comtrace equivalence there exists a sequence of equivalent step sequences
$(w_i)_{i=1\ldots n}$ such that $u=w_1$, $w_i\sim_\Theta w_{i+1}$, 
and $w_n=v$. In this sequence there has to exist an element $w_i$ where
the considered occurrences of indivisible steps were for the last time in the
same order as in $u$ ($w_i=w'_iX_iY_iw''_i$ and $w_{i+1}=w'_{i+1}Z_iw''_{i+1}$ and
$A\subseteq X_i$ and $B\subseteq Y_i$). Hence $A\times B\subseteq\Ser$.
Moreover, there exists an element $w_j$ where the considered occurrences
occur for the
first time after $w_i$ in the same order as in $v$ ($w_j=w'_jZ_jw''_j$ and 
$w_{j+1}=w'_{j+1}X_jY_jw''_{j+1}$ and $A\subseteq X_j$ and $B\subseteq Y_j$).
Hence also $B\times A\subseteq\Ser$. Therefore $(A,B)\in\widehat{\Ind}$,
which gives a contradiction and completes the proof.
\end{proof}

\begin{prop}\label{p:leadout}
Let $[w]$ be a comtrace over $\Theta$. Then
\[\widehat{w}\in[w]\]
and
\[\widehat{w}=w ~~\Leftrightarrow~~ w\in indiv(\tau).\] 
\end{prop}
\proof
Let $w=A_1\ldots A_n$. 
By the definition of the operator \;$\widehat{ }$\,, we get 
$$\widehat{w}=minlex(A_1)\ldots minlex(A_n).$$
Since $minlex(A_i)\equiv_{\Theta}A_i$ we get $\widehat{w}\equiv_\Theta w$, 
so $\widehat{w}\in [w]$.

Since $|minlex(A_i)|\geq 1$ and $|minlex(A_i)|=1$ if and only if 
$minlex(A_i)=\widehat{(A_i)}=A_i$ we conclude that
\[\widehat{w}=w ~~\Leftrightarrow~~ \forall_i\;\widehat{A_i}=A_i.\] 
By Lemma \ref{l:division}, $\widehat{(A_i)}=A_i$ if and only if $A_i$ is indivisible.
Hence $\widehat{w}=w$ if and only if all $A_i$ are indivisible and
\[\widehat{w}=w ~~\Leftrightarrow~~ w\in indiv(\tau).\eqno{\qEd}\]
\medskip

\noindent As an immediate corollary of Theorem \ref{t:trace4comtrace} and
Proposition \ref{p:leadout}, we can observe that 

\begin{cor}\label{c:correspondence}
There is a one to one correspondence between the comtraces over 
comtrace alphabet $\Theta=(\Sigma,\Sim,\Ser)$
and traces over concurrent alphabet 
$\Psi=(\widehat{\stepalph},\widehat{\Dep})$ given by
the construction of the set of indivisible steps 
and dependence relation on them. 
\[\begin{array}{ccc}
\tau & \widehat{\tau} & \tau '\\
over & \xlongrightarrow{\widehat{ }} \;\;over\; \longleftrightarrow & over\\
\Theta & \widehat{\Theta} & \Psi
\end{array}\]
\end{cor}

One can consider using the above correspondence to apply the 
methods of enumerating all traces of a given size~\cite{MikPiaSmy11} to enumerate 
comtraces of a given size.

\section{Projection Representation of Comtraces}

In the trace theory employing projections onto the cliques of the graph
of dependence relation
(see also \cite{Shi85}) turned out to be a very useful tool.
We now extend this notion in the case of the binary and unary cliques only 
(see also \cite{Mik08}), to define the projection representation of comtraces.
In the case of traces, we have only two kinds of relationships between actions.
As independent actions may be executed in any order (or together in case of
step semantics) one can focus on the order implied by the dependence relation.

In the case of comtraces, the situation is more complicated. However, once more we can
ignore independent actions and store information about the other three types
of relations (dependency, weak dependency and strong simultaneity). 
Once more, it is sufficient to store the information in the form of sequences.
In the case of strong simultaneity, however, we need to add a special symbol
$\perp$ that separates the situations of sequential and simultaneous execution
of pairs of actions being considered.

Let $a,b\in\Sigma$ and $(a,b)\notin\Ind$ (possibly $a=b$). 
For each such pair we define the
projection function $\Pi^\perp_{a,b}:\stepalph^*\rightarrow (\Sigma\cup\{\perp\})^*$ 
as follows. First, for a step $A\in\stepalph$ we have
\[
\Pi^\perp_{a,b}(A)=\left\{
\begin{array}{lll}
\epsilon & \text{ for } & \{a,b\}\cap A=\emptyset\\
a & \text{ for } & a\in A \wedge b\notin A\\
ba & \text{ for } & \{a,b\}\subseteq A \wedge (a,b)\in\Wdp\\
ab & \text{ for } & \{a,b\}\subseteq A \wedge (b,a)\in\Wdp\\
\perp & \text{ for } & \{a,b\}\subseteq A \wedge (a,b)\in\Ssm\\
\end{array}
\right.
\]

Note that there is a straightforward symmetry, namely for all $(a,b)\notin\Ind$
the equation $\Pi^\perp_{a,b}=\Pi^\perp_{b,a}$ holds.
Moreover, according to the definition, we have
$\Pi^\perp_{a,a}(A)=\epsilon$ if $a\notin A$ and $\Pi^\perp_{a,a}(A)=a$ if $a\in A$.
Then, for a step sequence $w=A_1A_2\ldots A_n$ we have
\[\Pi^\perp_{a,b}(w)=
\Pi^\perp_{a,b}(A_1)\circ\Pi^\perp_{a,b}(A_2)\circ\ldots\circ\Pi^\perp_{a,b}(A_n).\]

\begin{thm}\label{t:prrepr}
Let $w,u$ be step sequences over a comtrace alphabet 
$\Theta=(\stepalph,\Sim,\Ser)$.
Then $w\equiv_\Theta u \;\Leftrightarrow\; 
\forall_{(a,b)\notin\Ind}\;\Pi^\perp_{a,b}(w)=\Pi^\perp_{a,b}(u)$.
\end{thm}
\begin{proof}
$\Rightarrow :$\\
We first prove that \[w\equiv_\Theta u \;\Rightarrow\; 
\forall_{(a,b)\notin\Ind}\;\Pi^\perp_{a,b}(w)=\Pi^\perp_{a,b}(u).\]
According to the definition of comtrace equivalence, it is sufficient to prove
the statement in the case of equivalent step sequences $w=A$ and $u=BC$.
Let $a,b\in A$. We consider all but one of the possible relationships of 
these actions (the remaining case is that of independence).\\

Case 1: $(a,b)\in \Dep$.

Since actions $a$ and $b$ occur simultaneously in the step $A$, this is impossible.\\

Case 2: $(a,b)\in \Ssm$.

Since actions $a$ and $b$ are strongly simultaneous, Lemma \ref{l:division}
shows that they both have to occur in step $B$ or $C$. It means that
\[\Pi^\perp_{a,b}(BC)=\Pi^\perp_{a,b}(B)\Pi^\perp_{a,b}(C)=\perp\epsilon=\Pi^\perp_{a,b}(A)\]
or
\[\Pi^\perp_{a,b}(BC)=\Pi^\perp_{a,b}(B)\Pi^\perp_{a,b}(C)=\epsilon\perp=\Pi^\perp_{a,b}(A).\]
\mbox{}

Case 3: $(a,b)\in \Wdp$.

Since $B\times C\subseteq\Ser$, it is impossible that $b\in C$ and $a\in B$.
If they both belong to one step, we have 
\[\Pi^\perp_{a,b}(BC)=\Pi^\perp_{a,b}(B)\Pi^\perp_{a,b}(C)=(ba) \epsilon=\Pi^\perp_{a,b}(A)\]
or
\[\Pi^\perp_{a,b}(BC)=\Pi^\perp_{a,b}(B)\Pi^\perp_{a,b}(C)=\epsilon (ba)=\Pi^\perp_{a,b}(A)\]
while belonging to the different steps (namely $b\in B$ and $a\in C$) gives
\[\Pi^\perp_{a,b}(BC)=\Pi^\perp_{a,b}(B)\Pi^\perp_{a,b}(C)=ba=\Pi^\perp_{a,b}(A),\]
which completes the first part of the proof.\\

\noindent $\Leftarrow :$\\
Now, let us assume that we have two step sequences $u,v\in\stepalph^*$ and
\[\forall_{(a,b)\notin\Ind}\;\Pi^\perp_{a,b}(v)=\Pi^\perp_{a,b}(u).\] 
Without loss of
generality we can assume that $u=Au'$ is in the lexicographical canonical 
form and
$v$ consists of indivisible steps only. 
We claim that then there exist $v',v''\in\stepalph^*$ such that
$v=v'Av''$, no action occurring in
$A$ occurs in $v'$ and $A\times alph(v')\subseteq\Ind$.

Directly from the definition of the projection representation we see that all 
projections onto the subalphabets containing actions from the indivisible step
$A$ start with
the actions contained in $A$. More precisely, if $a,b\in A$ then 
$\Pi^\perp_{a,b}(u)$ starts with $ab$, $ba$ or $\perp$, depending on the relation 
between $a$ and $b$. If $a\in A$ and $b\notin A$ however, $\Pi^\perp_{a,b}(u)$ 
starts with a single action $a$. 

Let $v'$ be the longest prefix of $v$ such that $alph(v')\cap A=\emptyset$ 
and $v=v'Bv''$. Obviously, all 
projections onto the subalphabets containing actions from the step
$A$ are equal for $v$ and $Bv''$.
Moreover, from the 
definition of the indivisible step, between every two actions $a,b$ contained 
in $A$ there is a sequence of pairwise different actions $a=a_1,\ldots,a_n=b$
contained in $A$ such that for every $i<n$ we have $(a_{i+1},a_i)\in\Sin$.
It means that for every such a pair of consecutive actions we have
$\Pi^\perp_{a_i,a_{i+1}}(B)=a_ia_{i+1}$ if $(a_{i+1},a_i)\in\Wdp$ or 
$\Pi^\perp_{a_i,a_{i+1}}(B)=\perp$ if $(a_{i+1},a_i)\in\Ssm$. 
Nevertheless, if $a_{i+1}$ is in $B$ then also $a_i$ have to be in $B$. 
Otherwise $\Pi^\perp_{a_i,a_{i+1}}(B)$ would start with $a_{i+1}$.
This proves that, since $A\cap B\neq\emptyset$, $A\subseteq B$.
Using similar arguments, we can see that since $B$ is indivisible, no other
action may occur in $B$ and $A=B$.

It remains to be shown that $A\times alph(v')\subseteq\Ind$. Let $a\in A$ and 
$c\in alph(v')$. Clearly, $c\notin A$ from the definition of sequence $v'$. 
In the step sequence $v$ the action $c$ appears before 
action $a$ so, if they are not independent, $\Pi^\perp_{a,c}(v)=\Pi^\perp_{a,c}(u)$ starts
with $c$. But $a\in A$ and $c\notin A$, and so $\Pi^\perp_{a,c}(u)$ starts with $a$.
This contradicts our assumption that $a$ and $c$ are not independent and 
proves that $v\equiv_\Theta Av'v''$. Repeating the above reasoning, we obtain that
$u$ is the lexicographical canonical form of $v$ which ends the second part of 
the proof.
\end{proof}

\noindent The projection representation of a comtrace $\tau$ is a function 
$\Pi^\perp_\tau:(\Sigma\times\Sigma)\setminus\Ind\rightarrow (\Sigma\cup\{\perp\})^*$, 
given by $\Pi^\perp_\tau(a,b)=\Pi^\perp_{a,b}(\tau)$.
Moreover, any function 
$\Pi^\perp:(\Sigma\times\Sigma)\setminus\Ind\rightarrow (\Sigma\cup\{\perp\})^*$ 
is called a \emph{projection set}.
Clearly, not every projection set is a projection representation 
of a comtrace. In the next section, we give a procedure that decides whether 
a given projection set is a projection representation of a comtrace.
Moreover, if the answer is positive, the procedure computes 
a representative of such a comtrace.

First, however, we provide the algorithm computing projection representation of a comtrace. This algorithm comes
directly from the definition. 
However, to say anything about the time complexity of the algorithm,
it is important to discuss the data
structures which might be used by this algorithm. At the beginning, let us
consider the input. We get a comtrace alphabet $\Theta$ which consists of
the alphabet $\Sigma$ of size $k$ and two relations, $\Sim$ and $\Ser$, of size
at most $k^2$ each. We also get a step sequence $w$ which steps consist of $n$ 
occurrences of atomic actions (elements of $\Sigma$) all together. As a result, we obtain
the set of at most $k^2$ sequences (projections onto specified subalphabets).

We process the step sequence $w$ step by step, which means that the algorithm
is online (i.e. during the computation we achieve correct results for each proper
prefix of $w$). The processing of a single step is done according to
the definition of projections onto the pairs in the specified relation. 
It is worth carrying out some
preprocessing and, for every action, compute the list of all subalphabets in
which it may occur. By storing, for every computed projection, the number 
of the step when it was most recently updated, 
we avoid problems with the special cases of relations $\Wdp$ and
$\Ssm$ (in these cases two rather than one action may be added to one sequence 
while processing a single step). 

\begin{prop}\label{p:al1comp}
The procedure of computing $\Pi^\perp_\tau$ from a step sequence $w\in\tau$ has the 
time and memory complexity of $O(nk)$.
\end{prop}
\begin{proof}
The proof is straightforward. The algorithm is naturally divided into $n$
stages grouped by steps of input step sequence. In each stage we process
a single action and add it to at most $k$ sequences updating at most $k$
counters. Hence each stage can be done in the time
linearly proportional to the size of the alphabet.
Therefore whole procedure has the time complexity of $O(nk)$.
\end{proof}

\begin{thm}
Testing comtrace equivalence can be done in the time complexity of $O(nk)$.
\end{thm}
\proof
Notice that the output of procedure discussed in Proposition \ref{p:al1comp}
has also memory complexity of $O(nk)$. 
Hence for two step sequences we can compute their 
projection representations and compare them sequence by sequence.
\qed

\subsection{Reconstructing Step Sequence from Projection Set}

The idea of constructing a step sequence from a projection set is
based on revealing the first possible step whose projection
representation would form a set of prefixes of a given projection set.
At first, we identify the set of all possible elements of such a step.
We do it in two stages. We first identify the set of conditionally possible actions,
i.e. those actions whose first occurrences are the first (or in particular
situations the second) actions in all projections, where they could appear.
Note that we treat the special symbol $\perp$ as a pair of proper actions, so its occurrence means that both actions might be conditionally possible.
After this identification, we remove actions that cannot satisfy
some of the necessary conditions. These conditions are related to the cases when
the considered action appears as the second action in some sequences connected
with the weak dependence relation or are verified positively because of 
the special symbol $\perp$.

As a result of the first stage, we obtain the set of all actions that may appear
in the first step of the constructed sequence. The second stage consists of dividing
this set into indivisible steps and combining those indivisible steps
into one of the allowed steps. The result is obtained by taking advantage of 
the weak dependence relation inside the set of indivisible steps. It is
similar to the ideas behind the proof of Lemma~\ref{l:division}.
Let us look into the details of the proposed procedure.

Recall that by $\pref{k}(w)=a_1\ldots a_k$ we denote the k-prefix of $w$.
Let $\Pi^\perp$ be a projection set. We say that an action $a\in\Sigma$ is
\emph{conditionally possible} for projection set $\Pi^\perp$ if and only if
for all $b\in\Sigma$ the following implications are satisfied:\\

\begin{iteMize}{$\bullet$}
\item $(a,b)\in\Dep \Rightarrow \pref{1}(\Pi^\perp(a,b))=a$
\item $(b,a)\in\Wdp \Rightarrow \pref{1}(\Pi^\perp(a,b))=a$
\item $(a,b)\in\Wdp \Rightarrow \pref{1}(\Pi^\perp(a,b))=a 
\vee \pref{2}(\Pi^\perp(a,b))=ba$
\item $(a,b)\in\Ssm \Rightarrow \pref{1}(\Pi^\perp(a,b))=a 
\vee \pref{1}(\Pi^\perp(a,b))=\perp$
\end{iteMize}

We denote all conditionally possible actions as $cpa$ and define the
relation $cnd\subseteq\Sigma\times\Sigma$, which describes the conditions
that must be satisfied. 
Only in situations where 
\[(a,b)\in\Wdp \wedge \pref{2}(\Pi^\perp(a,b))=ba\]
or 
\[(a,b)\in\Ssm \wedge \pref{1}(\Pi^\perp(a,b))=\perp\] 
we say that the
existence of action $b$ in the constructed step is a necessary condition 
for the presence of action $a$ in this step, which is denoted by $(a,b)\in cnd$.

We exclude conditionally possible actions with conditions impossible to satisfy
to form the set of possible actions. Any action $a\in\Sigma$ that is not
conditionally possible in $\Pi^\perp$ is \emph{impossible} in $\Pi^\perp$. Moreover, any 
action $a$ conditionally possible under impossible condition 
(i.e. $(a,b)\in cnd$ and $b$ is impossible) is also impossible. Formally, the set 
of impossible actions for the projection function $\Pi^\perp$ is the smallest set 
$imp$ that satisfies the following conditions:

\begin{iteMize}{$\bullet$}
\item $\Sigma\setminus cpa\subseteq imp$
\item $b\in imp\wedge (a,b)\in cnd\Rightarrow a\in imp$
\end{iteMize}

Let $M(\Pi^\perp)$ be the set of actions which are not impossible (which means that they are possible) for projection set $\Pi^\perp$.
The next operation is to choose a subset of $M(\Pi^\perp)$ which could be a first step
of the reconstructed step sequence. 
To do so we take a sequential trace over $\widehat{\stepalph}$, 
given by the step sequence $\widehat{M(\Pi^\perp)}$ (see Corollary \ref{c:correspondence}).
Note that for any $a\in M(\Pi^\perp)$ we have $\#_a(\widehat{M(\Pi^\perp)})\leq 1$.
We take any nonempty trace prefix $B_1\ldots B_n$ of step sequence $\widehat{M(\Pi^\perp)}$
and set $B=\bigcup_{i}\;B_i$ as a requested step.
The procedure just described is justified by the following facts:

\begin{prop}
Let $\Pi^\perp$ be a projection set over a comtrace alphabet $\Theta$ and $B\subseteq 
M(\Pi^\perp)$ a set of actions constructed according to the procedure described above.

If $b\in B$ and $a\in M(\Pi^\perp)$ then
\[(b,a)\in\Sin^*~~\Longrightarrow~~[a]_{\equiv_{M(\Pi^\perp)}}\subseteq B.\]
\end{prop}
\begin{proof}
Let $\widehat{M(\Pi^\perp)}=B_1\ldots B_n$, where all the $B_m$'s are indivisible.
Since $a,b\in M(\Pi^\perp)$ there exist $1\leq p,q\leq n$ 
such that $a\in B_p$ and $b\in B_q$.
By Lemma \ref{l:division} $[a]_{\equiv_{M(\Pi^\perp)}}=B_p$.
By Corollary \ref{c:correspondence}, $\widehat{M(\Pi^\perp)}$ forms a sequential trace 
over $\widehat{\stepalph}$. 
In the above procedure we use one of trace prefixes of $\widehat{M(\Pi^\perp)}$,
taking $B$ as the union of all indivisible steps (actions of $\widehat{\stepalph}$)
contained in this prefix.
Hence $B_q\subseteq B$. 
If $p=q$ we have that $B_p=B_q$ and $B_p=[a]_{\equiv_{M(\Pi^\perp)}}\subseteq B$.
Let us consider the case $B_p\neq B_q$.
It is sufficient to prove that $B_p$ occurs before $B_q$ 
in all trace prefixes of $\widehat{M(\Pi^\perp)}$.

Since $(b,a)\in\sin^*$, there exists a sequence of actions $b=c_1\ldots c_k=a$
such that $(c_i,c_{i+1})\in\sin$ for every $0<i<k$.
Hence there exists a sequence of steps $u=C_1\ldots C_k$ such that $c_i\in C_i$.
Clearly, $C_i$ might be equal to $C_{i+1}$, for some $0<i<k$,
but surely $C_1\neq C_k$.
However, for distinct $i,j$ we have 
$(C_i,C_j)\in\widehat{\Dep}$.
Moreover, each $C_i$ is contained in $Alph(M(\Pi^\perp))$
and if $C_i$ occurs before $C_j$ in $u$, then it also has to occur before $C_j$
in $\widehat{M(\Pi^\perp)}$. 

If $C_i$ and $C_{i+1}$ are different, then during the division of the step 
$\widehat{M(\Pi^\perp)}$ (see Lemma~\ref{l:division}) they have to get to different parts 
(like steps $B$ and $C$ in Lemma \ref{l:division}).
Since $(c_i,c_{i+1})\in\Api$, it is impossible to have 
$C_i\times C_{i+1}\subseteq\Ser$.
This shows that their orders of occurring in $u$ and $\widehat{M(\Pi^\perp)}$ are reversed.
Moreover, this remains true for every sequence over $\widehat{\stepalph}$ equivalent
to $\widehat{M(\Pi^\perp)}$.
Finally, we conclude that what we have shown
applies not only to consecutive and distinct steps of $u$ but also
to all its distinct elements, including $C_1=B_q$ and $C_k=B_p$, which end the proof. 
\end{proof}

\begin{thm}
Let $w=A_1\ldots A_n$ be a step sequence, and $\Pi^\perp$ be the projection 
representation of $[w]$. Then
\[A_1\subseteq M(\Pi^\perp).\]
\end{thm}
\begin{proof}
Since $\Pi^\perp$ is the projection representation of $[w]$, 
for all $(a,b)\notin\Ind$ 
we have 
\[\Pi^\perp(a,b)=\Pi^\perp_{a,b}(w)=\Pi^\perp_{a,b}(A_1\ldots A_n).\]
Hence all actions contained in $A_1$ are conditionally possible.
Moreover, $(a,b)\in cnd$ means that $(a,b)\in\Wdp$ or $(a,b)\in\Ssm$.
In the first case, $\pref{2}(\Pi^\perp(a,b))=ba$, so $b\in A_1$.
Similarly, if $(a,b)\in\Ssm$ then $\pref{2}(\Pi^\perp(a,b))=\perp$, so $b\in A_1$.

Since $(a,b)\in cnd$ and $a\in A_1$ implies $b\in A_1$, and $A_1\subseteq cpa$,
we conclude that $A_1\cap imp=\emptyset$.
This proves that $A_1\subseteq M(\Pi^\perp)$.
\end{proof}

As a result, we can extract step $B$ from $\Pi^\perp$. 
The extraction function
\[extr:((\Sigma\times\Sigma\setminus\Ind)^*\rightarrow (\Sigma\cup\perp)^*)\times 
\stepalph\rightarrow((\Sigma\times\Sigma\setminus\Ind)^*\rightarrow (\Sigma\cup\perp)^*)\] 
for projection set $\Pi^\perp$ and set $B\subseteq M(\Pi^\perp)$ constructed using the 
procedure described above is defined as:

\[extr(\Pi^\perp,B)(a,b)=\left\{
\begin{array}{lcl}
\Pi^\perp(a,b) & \text{ for } & |\{a,b\}\cap B|=0\\
\suff{2}(\Pi^\perp(a,b)) & \text{ for } & |\{a,b\}\cap B|=1\\
\suff{2}(\Pi^\perp(a,b)) & \text{ for } & |\{a,b\}\cap B|=2 \wedge (a,b)\in\Ssm\\
\suff{3}(\Pi^\perp(a,b)) & \text{ for } & |\{a,b\}\cap B|=2 \wedge (a,b)\in\Wdp\cup\Wdp^{-1}\\
\end{array}
\right.
\]

\begin{exa}
Let us consider the comtrace $\tau$ from Example~\ref{e:comtrace}.

The projection representation of $\tau$ (omitting projections
to the unary subalphabets), grouped by the types of relation between the elements of 
subalphabets on which we project are:
\[
\begin{array}{lll}
\Dep: \;&\; \Pi^\perp_\tau(b,d)=db \;&\; \Pi^\perp_\tau(c,d)=d\\
\Ssm: \;&\; \Pi^\perp_\tau(a,c)=a &\\
\Wdp: \;&\; \Pi^\perp_\tau(c,b)=b \;&\; \Pi^\perp_\tau(d,a)=da
\end{array}
\]
The set of conditionally possible actions for $\Pi^\perp_\tau$ is $\{a,d\}$, while
$(a,d)\in cnd$. Every conditionally possible action is also possible,
and so $M(\Pi^\perp_\tau)=\{a,d\}$. This gives the set of two indivisible steps $(a)$ and
$(d)$ and, finally, two steps that may appear as the first step of the constructed
sequence: $(d)$ and $(ad)$.
\qed
\end{exa}

\begin{thm}\label{t:extrcorrectness}
Let $\Pi^\perp_\tau$ be the projection representation of a comtrace $\tau$, and 
$M(\Pi^\perp)$ be a maximal possible step of $\Pi^\perp_\tau$. For every allowed set 
$B\in\stepalph$, we have
\[\tau = B\circ \sigma, \text{ where }\Pi^\perp_\sigma=extr(\Pi^\perp_\tau,B).\]
\end{thm}
\begin{proof}
By the Theorem \ref{t:prrepr} it is sufficient to prove that
$\Pi^\perp_\tau = \Pi^\perp_{B\circ\sigma}$. In other words, we have to show
that for all $(a,b)\notin\Ind$, we have 
$\Pi^\perp_\tau(a,b)=\Pi^\perp_B(a,b)\circ\Pi^\perp_\sigma(a,b)$.

The proof can be split in a natural way into three parts, depending on the type of 
relation between the actions being considered. Let us examine 
the projections onto $(a,b)\in\Dep$.
We have $\Pi^\perp_B(a,b)$ that is equal to the first action of $\Pi^\perp_\tau(a,b)$ if 
$|B\cap\{a,b\}|=1$, and to $\epsilon$ otherwise.
In both cases $\Pi^\perp_B(a,b)\circ \Pi^\perp_\sigma(a,b)=\Pi^\perp_\tau(a,b)$.

Almost the same proof works for the remaining two cases, when $(a,b)\in\Wdp$ or
$(a,b)\in\Ssm$.
\end{proof}

By suitably using the extraction function, we can compute any representative of
a comtrace $\tau$. In particular, similarly to the case of canonical forms,
we can do this using a maximal or minimal strategy. In the maximal strategy, we
always take the whole set $M(\Pi^\perp)$ and, as a result, 
we obtain Foata canonical form of the original comtrace. 
In the minimal strategy, we take the first step of the step sequence
$\widehat{M(\Pi^\perp)}$ and obtain the lexicographical canonical form.

The algorithm reconstructing a step sequence from a projection representation
of a comtrace follows the notions defined above. From the technical point of
view, some concrete decisions concerning data structures are worth noticing. 
The whole 
algorithm can be divided into stages. In each stage we compute a set of allowed 
steps, choose one, and extract it from the projection set. The procedure is repeated
until a projection set $\Pi^\perp_i$ or computed set $M(\Pi^\perp_i)$ become empty. 
In the first case, it returns a step sequence consisting of $n$ occurrences of 
actions. In the second
case, the algorithm returns that an input is not a projection representation of
a comtrace.

A single stage starts from computing the set of conditionally possible actions
and the relation $cnd$ describing the conditions. A good idea is to preprocess,
for every action, a list of pointers which helps to
investigate only the projections related to this action.
Doing so, we can check conditional possibility in the time linearly dependent on
the size of alphabet, denoted by $k$. Simultaneously, we build the 
directed graph of conditions. In the time linearly dependent on the number of arcs
in this graph, we remove from the set of conditionally possible actions all
impossible ones (browsing, using DFS, all paths which begin in vertices which 
are not conditionally possible). 

In the next phase, we compute a vertex induced subgraph of the $\Api$ relation
that contains all possible actions and, once more using DFS, we compute a graph
of its strongly connected components 
(called \emph{condensation graph}~\cite{Deo74}). 
The condensation graph is an acyclic directed graph of the partial order 
of the sequential trace associated with $\widehat{M(\Pi^\perp)}$. 
We choose an arbitrary upper set of the condensation graph,
that corresponds to the trace prefix of $\widehat{M(\Pi^\perp)}$. 
To obtain Foata canonical form, we take the maximal 
upper set by choosing the whole condensation graph. If we wish to
obtain the lexicographical canonical form, we should choose the 
$\widehat{\leq}$-smallest allowed step. To compute it, we may consider only
the maximal elements of provided condensation graph.
They correspond to the elements of $\widehat{\stepalph}$ which may be placed
at the first positions in the sequential trace $\widehat{M(\Pi^\perp)}$.

In the last phase, we need to extract the chosen allowed step. We do it
according to the definition of the extraction operation. During this phase, we
can once more use the precomputed lists of pointers.

\begin{prop}
Projection set $\Pi^\perp$ is the projection representation of a comtrace if and only if
the procedure described above ends with the empty projection set.
\end{prop}
\begin{proof}
We give the proof only for the case when the maximal strategy is used.
Note that the input data is finite and the procedure stops when the set
$M(\Pi^\perp)$ is empty for the remaining set of words. From Theorem
\ref{t:extrcorrectness} we deduce that if the remaining projection set
is empty then the input is the projection representation of the constructed 
comtrace.
Suppose that we have nonempty projection set $\Pi^\perp$ that is a
projection representation of comtrace $\tau$ and empty set 
of allowed actions. Let us consider an arbitrary step sequence
$u=A_1\ldots A_n$ that is contained in $\tau$, and an arbitrary action $a$ 
contained in $A_1$.
Then, by the definition of projection representation, the action $a$
has to be possibly allowed. This proves that $A_1\subseteq cpa$.
Moreover, since in any projection before, or simultaneously with, $a$
may occur only other action from the step $A_1$, 
if the existence of action $b$ is a necessary condition for the
presence of action $a$ (i.e. $(a,b)\in cnd$), then $b$ is also an
element of $A_1$. Therefore, none of the actions from step $A_1$ is impossible,
which contradicts the emptiness of the set of allowed actions and 
ends the proof.
\end{proof}

\begin{thm}
The procedure of computing canonical forms from a projection representation of
a comtrace has the time complexity of $O(nk^2)$.
\end{thm}
\proof
The procedure consist of at most $n$ stages. In each part, we carry out
some operations
on at most $k^2$ lists and graph of size $k^2$. All graph operations, including
computing the compensation graph and choosing minimal or maximal upper set are
linear in the size of graph. This gives an overall time complexity of $O(nk^2)$.
\qed

\subsection{Traces as a subclass of comtraces}

In Section 1 we defined EN-systems as a special
case of ENI-systems without inhibitors and with the sequential semantics. 
We also introduced traces as a model of the causal behaviour of EN-systems.
In this section, we show what kind of comtraces are directly related to 
systems without inhibitors. 

A comtrace alphabet $\Theta=(\Sigma,\Sim,\Ser)$ with the empty relation $\Api$ is called \emph{\tracelike comtrace alphabet}. Moreover, comtraces over this alphabet are called \emph{\tracelike comtraces}.
The radicalism of such comtraces means that
the actions may be only dependent or independent, 
hence they behave exactly like step traces.
Later in this section we discuss some properties of this subclass.

\begin{prop}\label{l:singletons}
Let $\tau\in\stepalph^*$ be a \tracelike comtrace and $w\in indiv(\tau)$. Then each step of $w$ is a singleton.
\end{prop}
\begin{proof}
The proof is straightforward. Notice that since the relation $\Api$ is empty,
every action $a$ of every step $A\in\stepalph$ forms an indivisible step.
Hence all indivisible steps are singletons, which ends the proof.
\end{proof}

\begin{cor}\label{c:tracelikealph}
Let $\Theta=(\Sigma,\Sim,\Ser)$ be a \tracelike comtrace alphabet. 
Then 
\[lex(\widehat{\stepalph})=\Sigma.\]
\end{cor}

Note that since the relation $\Api$ is empty and all steps are
singletons, for all steps $A,B\in\widehat{\stepalph}$ we have 
$(A,B)\in\widehat{\Dep}$ if and only if $(lex(A),lex(B))\in\Dep$.

Using Theorem \ref{t:trace4comtrace} and Lemma \ref{c:tracelikealph} we can 
associate an alphabet of indivisible steps $\widehat{\stepalph}$ with 
$\Sigma$ and
\tracelike comtrace $\tau$ over a comtrace alphabet $\Theta=(\Sigma,\Sim,\Ser)$
with a step trace $\sigma$ over the concurrent alphabet
$\Psi=(\Sigma,\Ind)$. 
We say that such a step trace $\sigma$ is a \emph{trace representation} of 
a \tracelike comtrace $\tau$.
The following facts show this correspondence in details. 

\begin{prop}
Let $\Theta=(\Sigma,\Sim,\Ser)$ be a \tracelike comtrace alphabet.
Then a set $A\subseteq\Sigma$ is a step in $\Theta$ if and only if 
$A$ is a step in $\Psi=(\Sigma,\widehat{\Dep})$.
\end{prop}
\begin{proof}
It is sufficient to prove that $\Sim=\Ind$.

Indeed, since $\Api$ is empty, we have $\Sim\setminus\Ser=\emptyset$,
hence by $\Ser\subseteq\Sim$ we get $\Sim=\Ser$.
Recall that $\Sim$ is symmetric, and so is $\Ser$.
By the definition of relations in comtraces, 
\[\Ind=\Ser\cap\Ser^{-1}=\Ser=\Sim.\]
\end{proof}

\begin{thm}
Let $\tau$ be a \tracelike comtrace and $\sigma$ be its step trace representation.
Then $\Pi^\perp_\tau=\Pi_\sigma$.
\end{thm}
\begin{proof}
Let $\tau=[A_1\ldots A_n]=\sigma$. 
The relation $\Api$ is empty, so in the case of comtraces we consider
only projections to pair of actions that are dependent.
As a result, we conclude that the projections on the same pairs of actions 
are the same, no matter whether we consider comtraces or step traces, 
$\Pi^\perp_{a,b}(\tau)=\Pi_{a,b}(\sigma)$ for every $(a,b)\in\Dep$, 
hence $\Pi^\perp_\tau=\Pi_\sigma$.
\end{proof}

\begin{cor}
Let $\tau\in\stepalph^*$ be a \tracelike comtrace and 
$\sigma$ be its step trace representation. 
Their canonical forms (both lexicographical and Foata) are equal. 
\end{cor}

\begin{cor}\label{c:TLcorrespondence}
The correspondence between comtraces over 
$\Theta=(\Sigma,\Sim,\Ser)$
and traces over $\Psi=(\widehat{\stepalph},\widehat{\Dep})$ 
(see Corollary \ref{c:correspondence}) 
collapses in case of \tracelike comtraces to 
\[\begin{array}{ccc}
\tau & &\sigma\\
over & \longleftrightarrow & over\\
\Theta & &\Psi
\end{array}\;,\]
where $\tau=\sigma$ as sets of step sequences.
\end{cor}

\section{Summary and future work}

In this paper we presented a number of algebraic aspects of combined traces.
Similar algebraic tools were successfully used in the study of the 
Mazurkiewicz traces, a simpler model for capturing and analysing concurrent
behaviours.

In particular, we defined lexicographical canonical form of a comtrace
and its projection representation. We gave two simple algorithms which
generate these representations from arbitrary step sequence. Those
algorithms seem to have the potential to provide a base for the development of
solutions to some natural problems related to the comtrace theory, 
like model verification \cite{EspHel08,RodSchKho13}.
In particular, one can use them to design efficient methods for the enumeration
of all the representatives of a fixed comtrace, and the enumeration of all
comtraces of a given size.

Another interesting direction of further studies would be the notion of
recognisable and rational languages of combined traces. The projection
representation seems to be a good starting point in this area; in particular, if
one recalls Zielonka's asynchronous automata~\cite{Zie87} for traces.
Finally, the projection representation may find an application in another
important aspect of combined trace theory. A fair strategy of reconstructing
step sequences from a projection set might be useful as a starting point in
the theory of infinite combined traces.

\subsection*{Acknowledgments}
I would like to thank Maciej Koutny and anonymous reviewers for their constructive comments, which helped to improve this paper.\\
This research was supported by a fellowship funded by the
``Enhancing Educational Potential of Nicolaus Co\-per\-ni\-cus
University in the Disciplines of Mathematical and Natural Sciences''
Project POKL.04.01.01-00-081/10.

\bibliography{etykaInic}
\vspace{-30 pt}
\end{document}